\documentclass{article}

\usepackage[pdftex]{pict2e}
\usepackage{amsthm}
\usepackage{graphicx}
\usepackage{epstopdf}
\usepackage{indentfirst}
\usepackage[numbers]{natbib}
\usepackage{amstext}
\usepackage{enumitem}

\usepackage{stackrel}
\usepackage{mathpazo}
\usepackage{graphicx}
\usepackage[a4paper,bindingoffset=0.2in,%
            left=0.7in,right=0.7in,top=0.9in,bottom=1in,%
            footskip=.25in]{geometry}
\usepackage{amsfonts}
\usepackage{amssymb}
\usepackage{amsthm}
\usepackage{amsmath}
\usepackage{latexsym}
\usepackage{mathrsfs}
\usepackage{color}
\usepackage{dsfont}
\usepackage{url,stackrel,enumitem}	
\usepackage{hyperref}
\usepackage{amsfonts}
\usepackage{stmaryrd}
\usepackage{euscript}
\usepackage{amscd}
\usepackage{bm}
\usepackage{subfig}

\newtheorem{theorem}{Theorem}[section]
\newtheorem{lemma}{Lemma}[section]
\newtheorem{proposition}{Proposition}[section]

\newtheorem{example}{Example}[section]

\newtheorem{definition}{Definition}[section]

\newtheorem{hyp}{Assumption}[section]
\newcommand{\bhyp }{\begin{hyp} \rm }
\newcommand{\ehyp }{\end{hyp}}

\newcommand\I{\mathds{1}}

\newcommand{\pp}{\psi}

\newcommand{\E}{{\mathbb E}}

\newcommand{\Q}{{\mathbb Q}}
\newcommand{\EQ}{{\mathbb E}_{{\mathbb Q}}}

\newcommand{\bbP}{{\mathbb P}}
\newcommand{\bbR}{{\mathbb R}}

\newcommand{\cF}{\mathcal F}

\date{}

\begin{document}
\title{{\Large \bf  Valuation of Variable Annuities with Equity Protection Swaps \\ under Jumps and Default Risks} \vskip 35 pt }

\author{
Marek Rutkowski$^{a,b}$\thanks{%
\noindent\textbf{Acknowledgements.}
The research of M. Rutkowski and H. Xu was supported by the Australian Research Council Discovery Project DP200101550.}
\and
Huansang Xu$^{a,c}$
}

\maketitle

\begin{center}
\normalsize{\today} \\ \vspace{0.5cm}
\small\textit{$^{a}$ School of Mathematics and Statistics, University of Sydney, \\ Sydney, NSW 2006, Australia}\\
\small\textit{$^{b}$ Faculty of Mathematics and Information Science, Warsaw University of Technology, \\ 00-661 Warszawa, Poland}\\
\small\textit{$^{c}$ National University of Singapore, Department of Mathematics,\\ 21 Lower Kent Ridge Road, 119077.}
\end{center}


\begin{abstract}

This paper examines the valuation and hedging of standard equity protection swap (EPS) products proposed by Xu et al. \cite{Xu2024}. To account for financial crises and counterparty default risk, we develop pricing frameworks based on Merton’s jump-diffusion model and Szimayer’s independent random time default model, under which closed-form valuation formulas and put–call parity relations for European options are derived.
Hedging strategies for EPS products are analysed under jump and default risks. While static hedging remains effective in the absence of default, counterparty default risk leads to residual losses that cannot be fully hedged. These losses are quantified and used to define default-adjusted initial premiums under both Black–Scholes and jump-diffusion settings. Numerical results illustrate the effects of jump characteristics and default intensity on hedging costs and premiums, highlighting the importance of incorporating crisis and credit risks in EPS pricing and risk management.
\\[2mm]
 
\end{abstract}

\newpage

\section{Introduction} \label{sec6}

Rapid population ageing has increased longevity risk, heightening the demand for retirement products that balance market participation with downside protection. Traditional variable annuities (VAs) address this need through embedded guarantees, while more recent products such as registered index-linked annuities (RILAs) combine equity exposure with limited downside risk in a simpler structure \cite{Moening2021}. Building on this evolution, Xu et al. \cite{Xu2024} propose the equity protection swap (EPS), a flexible insurance-based derivative designed for superannuation members. Unlike VA riders or RILAs, the EPS is not an annuity itself but an add-on insurance product that provides tailored protection against portfolio losses with reduced structural complexity.

For clarity, we consider equity protection swap contracts, with or without variable annuities, that provide protection when the value of the underlying assets declines and collect insurance fees when the assets appreciate.
We now desire to analyze a more general setting in which the holding period of the reference portfolio is variable. In practice, both insurers and investors are highly concerned about financial crises, during which asset values may suddenly fall by a large amount. Although financial crises have been extensively studied and many theories have been proposed to explain their origins and prevention, there is still no consensus, and such crises continue to occur.

In practice, the valuation and hedging of EPS products do not occur simultaneously. The provider must quote a fair initial premium to the investor, after which the investor requires time to decide whether to enter the contract. Only once the contract is signed can the provider construct the hedging portfolio. As a result, the actual hedging cost generally differs from the initially quoted premium. Although the provider may state that the premium is time-dependent, a delay between contract execution and hedging is unavoidable. Consequently, perfect hedging is not achievable in real markets.
An even more adverse scenario arises when severe market turmoil prevents the construction of a hedging portfolio altogether. To capture the unobservable risks generated by such timing mismatches, we incorporate random jumps into the underlying asset dynamics.

In our framework, random jumps represent sudden and substantial equity losses and serve as a stylized way to model financial crises. Incorporating jumps into the dynamics of the EPS reference portfolio therefore provides a more realistic setting for hedging analysis.
Various jump models have been proposed in the literature. The most general class is the jump-diffusion model, in which asset prices evolve according to a diffusion process augmented by jumps whose sizes and distributions depend on the available information, as described by an appropriate $\sigma$-algebra. Kaushik \cite{Ka1993} introduced an early discrete-time framework for option valuation under jump-diffusion dynamics by superimposing multivariate jumps onto a binomial model, allowing for early exercise and arbitrary jump distributions. Later, Cont and Tankov \cite{CT2004} developed a non-parametric calibration approach for exponential Lévy (jump-diffusion) models using observed option prices.
Motivated by these contributions, we adopt the jump-diffusion model to study the impact of jumps during the holding period of standard EPS products.

Furthermore, we will discuss another important situation - credit risks. When a jump happens, it may cause a default event in the real market, especially for negative jumps, this will have a profound effect on the EPS provider's cash flow. Default risk, also known as default probability, is the likelihood that, in accordance with the conditions of the relevant debt instrument, a borrower will not make timely and complete payments of principal and interest. Default risk is one of the two elements of credit risk, along with loss severity. 
In this paper, we will consider a simple case of default - random time default event, according to Szimayer \cite{Szi2005}.
Under this model, the default event can be seen as jumps out of assets that the event risk is modelled by a random time $\tau$ which announces the event. A random time $\tau$ associated with an indicator process is used and the indicator process has a density that corresponds to the frequency of default. Szimayer \cite{Szi2005} studied the valuation of American options in the presence of external/non-hedgeable event risk by using a random time to represent the jump, or default actually. When the event occurs, the American option is terminated and a rebate is paid instead of the promised pay-off profile. We find that under this model, this random time default event can be regarded as a discount factor similar to the interest rate which is independent of the dynamics of assets. We can consider this random time default event as a jump-to-ruin situation and the valuation and hedging are easier under this model than under the jump-diffusion model.
As a remark, we need to mention that regular jumps can represent known dividends from equities which is also a meaningful study for the real market. 

In the following section, we describe the key features of the standard EPS product and give the definition of the structure of this product. 
Some simple examples of the standard EPS products - buffer EPS, floor EPS, and floor-cap EPS - are provided in Section \ref{sec6.0} according to Xu et al. \cite{Xu2024}. Then Section \ref{sec6.3+} is devoted to providing two different models and valuation formulas under these models. Considering two influence factors in this paper, we use Merton's jump-diffusion model \cite{Merton1976} to represent unknown financial crisis and the influence of the time difference between pricing and hedging, and Szimayer's independent random time default model \cite{Szi2005} to analyse the credit risks. We further provide the valuation formulas of European options and put-call parity under these two models.

In Section \ref{sec6.4+}, we analyse the hedging strategies under the consideration of financial crisis, cancellation risk, execution risk and credit risks. Firstly, hedging strategies for standard EPS products under the jump-diffusion model without consideration of default risks have been given, and these are actually the same as under the Vanilla case discussed in Xu et al. \cite{Xu2024}. We further discuss the situation that only the counter party of the EPS provider has default risk, and we consider the normal case without jumps for simplicity. Here we cannot find the static hedging strategy that can fully hedge the EPS product under this situation, and thus we find the EPS provider's cash flow after applying the static hedging strategy under the Vanilla case. We can see that if we assume only the counter party has the probability of default, the EPS provider will still face a potential loss after hedging and the "fair" initial premium defined under the Vanilla case is no longer "fair" for the EPS provider under the consideration of default risks. Fortunately, the potential loss is limited by the initial nominal principal of the EPS product and we can calculate the EPS provider's expected loss after hedging. After adding a default adjustment, we provide the definition of default adjusted initial premium when only the counter party with default risks, and we consider both the standard Black-Scholes model and Merton's jump-diffusion model.

Finally, we present numerical illustrations for the standard EPS product with the consideration of jumps and default events in Section \ref{sec6.5}. We start with the valuation of European options under the jump model in order to further analyse the influence of jumps on the hedging costs of the standard EPS products. We consider the jump size to satisfy a normal distribution with different parameters and different jump intensities, and these jumps are assumed to be negative which represents the financial crisis. Then we discuss the influence of different characteristics of jumps, such as jump intensity, jump size, and jump to ruin (default), on the hedging cost under the jump-diffusion model. Moreover, we consider the number of jumps in the holding period of the EPS products to be no more than one, exactly one, and around 20 times. We also consider the hedging costs with default risks independently and we use different default intensities to represent probabilities of default. Because the static hedging strategy cannot fully hedge the EPS provider's cash flow, we further present a numerical study for the default-adjusted initial premiums, whereas the underlying asset satisfies Merton's jump-diffusion model and Szimayer's independent random time default event is included in the modelling. Lastly, we give a conclusion in Section \ref{sec8} to summarise our findings and the importance of this study.

\section{Preliminary} \label{sec6.0}

In this section, we introduce the valuation framework for Equity Protection Swap (EPS) contracts and outline two jump-based models that will be specified later. We focus on standard EPS contracts, which have a single terminal payoff at maturity $T$. The payoff depends on the nominal principal $N_p$ and the simple rate of return of an underlying reference portfolio $S$ over the contract period $[0, T]$. The simple return is defined as $R_T = (S_T-S_0)/S_0 > -1$.

Following Xu et al. \cite{Xu2024}, the terminal cash flow from the perspective of the protection provider is given by $N_p \pp (R_T)$, where $\pp : \bbR \to \bbR$ characterizes the EPS payoff structure. We assume that $\pp$ is a continuous, non-decreasing, piecewise linear function satisfying $\pp (0)=0$. Unless stated otherwise, we normalize the nominal principal to one, i.e. $N_p=1$.

\begin{lemma} \label{lem0.2.1}
In a \textit{standard equity protection swaps} starting at time $0$ with the maturity date $T$ and the reference portfolio $S$,
the cash flow for the EPS provider is specified by the \textit{adjusted return} $\pp (R_T)$.
We have that $\pp (R) \leq 0$ for $R \in (-\infty ,0]$ and $\pp (R)\geq 0$ for $R \in [0,\infty)$. For $i=0,1,\dots,n$  and $j=0,1,\dots ,m$, the {\it protection leg} and the {\it fee leg} of an EPS are given by

\begin{align*}\label{eq0.2.1}
\pp^p(R)&:=\pp(R)\I_{\{R<0\}} =\Big(\sum_{k=0}^{i-1}p_{k+1}\Delta l_k+p_{i+1}(R-l_i)\Big)\I_{\{R \in (l_{i+1},l_i)\}}, \\
\pp^f(R)&:=\pp(R)\I_{\{R>0\}} =\Big(\sum_{k=0}^{j-1}f_{k+1}\Delta g_k+f_{j+1}(R-g_j)\Big)\I_{\{R \in (g_j,g_{j+1})\}},
\end{align*}

where $\Delta l_i=l_{i+1}-l_i<0$ and $\Delta g_j=g_{j+1}-g_j>0$ and recall that $g_0=l_0=0$.
We have $\pp (R_T)=\pp^p(R_T)+\pp^f(R_T)$, where the non-positive payoff $\pp^p(R_T)$ and the non-negative payoff $\pp^f(R_T)$ represent the provider's loss (protection payout) and gain (fee income) from an EPS, respectively.
\end{lemma}

{\bf Examples of Standard EPS Products}

We now introduce three standard EPS products with specific payoff structures. These products will serve as benchmark examples for the subsequent analysis. We first define the {\it buffer EPS}. Since a buffer is applied to both the protection and fee legs, this product may also be referred to as a double-buffer EPS.

\begin{definition} \label{BEPS}
{\rm The (double) \textit{buffer EPS} is obtained by taking $p_1=f_1=0$ (hence $\pp_B (R)=0$ for all $R \in [l_1,g_1]$ where $l_1<0$ and $g_1>0$) and some values for $p_2$ and $f_2$ in $(0,1]$. Hence the payoff profile at maturity $T$ of a buffer EPS for the provider equals}
\begin{align*}
\pp_B (R_T)& := \pp^{p,2}(R_T)+\pp^{f,2}(R_T) = p_2(R_T-l_1) \I_{\{R_T \in (-\infty,l_1)\}}+f_2(R_T-g_1)\I_{\{\{R_T\in (g_1,\infty)\}}\\
& = -p_2(l_1-R_T)^+ +f_2(R_T-g_1)^+.
\end{align*}
\end{definition}

The {\it floor EPS} is obtained by setting $p_2=0$, $p_1\in(0,1]$, $f_1=0$, and $f_2\in(0,1]$. In this case, losses above the level $l_1$ are fully covered by the provider, while losses below $l_1$ are not covered. As before, the extreme case $f_2=1$ corresponds to full participation in gains above the cap $g_1$.
By convention, the buffer is applied to the fee leg, as in the {\it buffer EPS}. This reflects the practical consideration that holders typically prefer to pay fees only when their gains over $[0,T]$ are sufficiently large.

\begin{definition} \label{FEPS}
{\rm The \textit{floor EPS} is specified by $p_2 =0, p_1 \in (0,1], f_1 =0, f_2 \in (0,1]$ and thus its
payoff profile for the provider equals}
\begin{align*}
\pp_F (R_T)&:= \pp^{p,3}(R_T)+\pp^{f,2}(R_T)=p_1 l_1\I_{\{R_T \in (-\infty,l_1]\}}-p_1(-R_T)^+\I_{\{R_T \in (l_1,0]\}}+f_2(R_T-g_1)^+
\\ &=-p_1 (-R_T)^++p_1(l_1-R_T)^++f_2(R_T-g_1)^+.
\end{align*}
\end{definition}

Finally, a {\it floor-cap EPS} can be viewed as an extension of the {\it buffer EPS}, with both a floor and a cap added on either side of the buffer. Although the most general specification involves three protection and three fee legs, for comparability we consider a simplified structure with complexity similar to the buffer and floor EPS defined above.

\begin{definition} \label{FCEPS}
{\rm The \textit{floor-cap EPS} is specified by $p_2 =0, p_1 \in (0,1], f_1 \in (0,1], f_2=0$ and thus its
payoff profile for the provider equals}
\begin{align*}
\pp_{FC} (R_T)&:= \pp^{p,3}(R_T)+\pp^{f,3}(R_T)
\\ &=p_1 l_1\I_{\{R_T \in (-\infty,l_1]\}}-p_1(-R_T)^+\I_{\{R_T \in (l_1,0]\}}+f_1 R_T^+ \I_{\{R_T \in (0,g_1]\}} + f_1 g_1 \I_{\{R_T \in (g_1,\infty)\}}
\\ &=-p_1 (-R_T)^+ +p_1(l_1-R_T)^+ +f_1 R_T ^+ -f_1(R_T-g_1)^+.
\end{align*}
\end{definition}

Furthermore, we introduce a general jump-diffusion framework to study hedging costs in the presence of jumps. Random default times are also discussed as a related extension of the jump setting.

\section{Modelling Framework} \label{sec6.3+}

As we already discussed in the introduction, the existence of a temporal discrepancy between EPS provider's hedging and pricing is beyond debate when we examine our EPS products in the actual market. The consumer will be informed of the fair starting premium before or close to the contract signing time, which is inconsistent with the hedging time of the EPS provider. It goes without saying that the true cost of hedging will differ from the reasonable initial premium that was previously disclosed to the client. Furthermore, there is an even worse scenario when the hedging portfolio cannot be built due to extreme market turbulence. In conclusion, we must take jumps into account because there is no perfect hedging in the real market.

Thus after defining the structure of standard EPS, we now introduce the jump model to the underlying asset valuation. Here we will give the definition of the general jump model - the jump-diffusion model, which we will use to discuss the influence of jumps on the pricing and hedging of standard EPS products in this paper. 

Besides jumps, cancellation risks and execution risks of the options should also included in the consideration of EPS products' hedging strategy, especially after the consideration of jumps. It is important to note that credit risks will have a significant impact on the hedging strategies and fair initial premiums of the standard EPS with consideration of jumps. In this paper, we will discuss the independent random time default event, which is a hazard process applied to the underlying asset. This independent default event can be seen as the simplest case of consideration of default risks, which is an external/non-hedgeable event risk. Although it may not be very consistent with the real market situation that the default event may related to jumps directly, it is an easy access to evaluate default risks and still meaningful.

Now we will start with the structure of the jump-diffusion model and then the independent random time default model. We will find the valuation formulas of the financial derivatives under these two models, which describe jumps and default events. The explicit pricing formulas of the European options will be given in this section that we will use for the hedging strategies later. 

\subsection{Jump-Diffusion Model} \label{sec6.1}

We consider the general model when analysing jumps, the jump-diffusion model, which was first introduced by Merton \cite{Merton1976}, in order to do further hedging strategies for Equity Protection Swap (EPS) with jumps. Jumps are related to the values of stocks directly in the jump-diffusion model, which means they represent instant huge capital gains or losses rather than default probabilities. Jump diffusion models are particular cases of exponential L{\'e}vy models in which the frequency of jumps is finite. They can be considered prototypes for a large class of more complex models such as the stochastic volatility plus jumps model.

Using the same algorithms of Merton \cite{Merton1976} and Kaushik \cite{Ka1993}, consider a market with a riskless asset (the bank account) and one risky asset (the stock) whose price at time $t$ is denoted by $S_t$. Then we have the definition of the stock process.

\begin{definition} \label{jump_S}{\rm
In a jump-diffusion model, the dynamics of the stock price $S_t$ is given as}

\begin{equation} \label{eq6.1.1}
dS_t = \mu S_{t^-}\,dt + \sigma S_{t^-}\,dW_t + S_{t^-}\,dJ_t ,
\end{equation}

{\rm where $W_t$ is a Brownian motion under the real market measure $\bbP$ and}
\begin{equation*}
J_t = \sum_{i=1}^{N_t} Y_i    
\end{equation*}
{\rm is a compound Poisson process where the jump sizes $Y_i$ are independent and identically distributed with distribution $F$ and number of jumps $N_t$ is a Poisson process with jump intensity $\lambda$.}
\end{definition}

The asset price $S_t$ thus follows geometric Brownian motion between jumps. Monte Carlo simulation of the process can be carried out by first simulating the number of jumps $N_t$, the jump times, and then simulating geometric Brownian motion on intervals between jump times.

Firstly, the stochastic differential equation \ref{eq6.1.1} has the exact solution according to It{\'o}'s formula under the real market measure $\bbP$:

\begin{equation*} 
S_t = S_0 \exp \big\{ \mu t + \sigma W_t - \frac{\sigma^2}{2}t + J_t \big\}.    
\end{equation*}

If we want to use the above model in the pricing formula, it is required that the asset grows at a risk-free rate, so that the discounted prices become a martingale under the risk-neutral measure. The logical thing that we need to do is find the risk-neutral drift $\widehat{\mu}$ and jump compensator $\mu_J$, then take the other parameters as risk neutral too. We have the following theorem:

\begin{theorem} \label{compensator}{\rm
If the stock satisfies dynamics given in Eq. \ref{eq6.1.1}, the jump compensator $\mu_J$ is given as follows under an equivalent measure, $\Q \sim \bbP$:}

\begin{equation*}
\mu_J = - \lambda \int \zeta F (d\zeta).    
\end{equation*}
\end{theorem}

\begin{proof}
Firstly,  we can get that the drift term $\mu$ in Eq. \ref{eq6.1.1} is given by the combination of a risk-neutral drift $\widehat{\mu}$ and a jump compensator $\mu_J$:

\begin{equation*}
\mu = \widehat{\mu} + \mu_J.    
\end{equation*}

To simplify the presentation, we henceforth assume zero dividends so that $\widehat{\mu} = r$, which is the risk-free rate.
By taking expectations of Eq. \ref{eq6.1.1} and using definition of risk-neutral drift $r$, we have

\begin{equation*}
\E(dS_t) = r S_t\,dt = \mu S_t\,dt + 
\lambda \Big\{\int \zeta F (d\zeta)\Big\} S_t dt    
\end{equation*}

where $F$ is the risk-neutral jump distribution. Then we can identify the jump compensator $\mu_J$ under the equivalent measure $\Q$.
\end{proof}

As a remark, jump-diffusion models are incomplete market models, since there is more than one equivalent martingale measure $\Q \sim \bbP$ under which the process of the discounted asset price becomes a martingale. In terms of the hedging portfolio, this is equivalent to saying that there is no way of building a completely risk-free hedging portfolio.

According to Theorem \ref{compensator}, by changing the drift of the Brownian motion and keeping the rest unchanged, we can get an equivalent measure $\Q \sim \bbP$. Under this measure $\Q$, the stock price satisfies

\begin{equation} \label{eq6.1.2}
S_t = S_0 \exp \big\{ (r - \frac{\sigma^2}{2} +\mu_J)t + \sigma W_t + J_t \big\}.    
\end{equation}

Before we find the pricing formula of European options for underlying assets with jumps, we need to define the forward price $P:= S_0 e^{rT}$, then we can get the characteristic function. 

\begin{definition} \label{characteristic}{\rm
With a stock $S$ and its forward price $P = S_0 e^{rT}$, define $x_t := \log(S_t/P)$, which is a L{\'e}vy process. The characteristic function of this L{\'e}vy process is given as follows}

\begin{equation} \label{eq6.1.3}
\phi_T(u):= \E (e^{iu x_T}).    
\end{equation}

{\rm This characteristic function $\phi_T(u)$ has the L{\'e}vy-Khintchine representation}

\begin{equation} \label{eq6.1.4}
\phi_T(u) = \exp \Big\{ iu \big(\mu_J - \frac{\sigma^2}{2} \big)T 
- \frac{u^2 \sigma^2}{2}T + T \int [e^{iu \zeta}-1] 
v(\zeta)\,d\zeta \Big\}.   
\end{equation}
\end{definition}

Now we make a typical assumption for the distribution of jump sizes $Y_i$, which is the normal distribution same as in the original paper by Merton \cite{Merton1976}. In Merton \cite{Merton1976} model, the L{\'e}vy density $v(\cdot)$ is given by

\begin{equation*}
v(\zeta) = \frac{\lambda}{\sqrt{2 \pi} \delta} \exp 
\bigg\{ -\frac{(\zeta - \alpha)^2}{2\delta^2} \bigg\}
\end{equation*}

where $\alpha$ is the mean of the log-jump size $\log J$ and $\delta$ the standard deviation of jumps. 

\begin{proposition}
By assuming the distribution of jump sizes $Y_i$ to be a normal distribution, the explicit characteristic function is displayed by

\begin{equation*}
\phi_T(u) = \exp \Big\{ iu wT - \frac{1}{2} u^2 \sigma^2 T + \lambda T
\big( e^{iu \alpha - \frac{u^2 \sigma^2}{2} - 1} \big) \Big\}
\end{equation*}

with
\begin{equation*}
w = -\frac{1}{2} \alpha^2 - \lambda \big( e^{iu \alpha - \frac{u^2 \sigma^2}{2} - 1} \big).   
\end{equation*}
Here $\alpha$ is the mean of the log-jump size $\log J$.
\end{proposition}

\subsection{Independent Random Time Default} \label{sec6.2}

Now we will introduce the independent random time default into our valuation of stock prices, which is a relatively simple model compared with the jump-diffusion model. We need to mention that the random time default event can be seen as a jump-to-ruin situation that describes default probabilities. Here, we model a random time $\tau$ which is associated with some indicator process to represent random default. Unlike the jump-diffusion model, these default events are independent of the value process of assets because we use a random time that is out of the financial market, thus it cannot be shown in the dynamics of assets.

Consider a market with a riskless asset (the bank account) and one risky asset (the stock) whose price at time $t$ is denoted by $S_t$. By using the same modelling structure in Szimayer \cite{Szi2005}, we now have the definition for the independent random time default event.

\begin{definition} \label{random_t}{\rm
For an independent random time default event under probability measure $\bbP$, the random time $\tau$ is associated with its indicator process $D$ (means default) via}

\begin{equation*}
D_t = \I_{\{\tau \leq t\}}, \quad \text{for} \quad 0 \leq t \leq T.   
\end{equation*}



{\rm The tacit underlying assumption that the event risk is an extraneous risk is now formalised: the indicator process $D$ is assumed to have a $\mathbf{F}$ - compensator admitting a density $\gamma^P$, which is predictable and strictly positive on $[0,\tau]$ under probability measure $\bbP$, i.e.,}

\begin{equation} \label{eq6.2.1}
M_t = D_t - \int_0^{t \wedge \tau} \gamma^P_u\,du, \quad \text{for} \quad 0 \leq t \leq T,    
\end{equation}

{\rm defines a $\mathbf{F}$-martingale with respect to $\bbP$.}
\end{definition}

Under this model, we consider $\tau$ is the first default time of a Cox process. In addition, the assumption that $\gamma^P$ is constant results in $\tau$ being an exponentially distributed stopping time, moreover, $\tau$ is independent of the financial market $(B,S,\mathbf{F})$. The independent holds of course also if intensity $\gamma^P$ is deterministic. Then we can find the following proposition by considering the distribution of random time.

\begin{proposition} \label{prop6.2.1}
Under the random time default model, we can find the probability of random time $\tau$ that has the following consequence

\begin{equation} \label{eq6.2.2}
P(\tau > t | \mathcal{F_T}) = \exp{\Big( -\int_0^t \gamma_u^P du \Big)}, \quad \text{for} \quad 0 \leq t \leq T,    
\end{equation}

and unconditionally that

\begin{equation*} 
P(\tau > t) = \E_P \Big\{ \exp{\Big(-\int_0^t \gamma_u^P du \Big)} \Big\}, \quad \text{for} \quad 0 \leq t \leq T.    
\end{equation*}
\end{proposition}

From the definition and proposition before, conditionally on the entire information of the financial market $\mathcal{F_T}$, the random time $\tau$ can be regarded as the first jump of an (inhomogeneous) Poisson process with intensity $\gamma^P$, or equivalently, conditioned on $\mathcal{F_T}$, we can view the stopping time $\tau$ as exponentially distributed with stochastic intensity $\gamma^P$.

\subsection{Valuation Formulation under Jump-Diffusion Model} \label{sec6.1.1}

After defining the jump-diffusion model and an independent random time default event, we can now find the valuation formula for derivatives with consideration of jumps and default risks.
As the random time default event is independent of the dynamics of the underlying asset, we will discuss the valuation under the jump-diffusion model first and add the influence of default risks later.
Moreover, we will focus on European options only and they will be used for hedging the standard EPS products.






{\bf A Valuation Formula for European Options} \label{sec6.1.2}

With characteristic function and valuation equation for European options, we can find an explicit solution under the jump-diffusion model according to Merton \cite{Merton1976}.

\begin{theorem} \label{jump_option}{\rm
For a contingent claim - European-style option - $Z^J(S,K,T)$ with strike $K$ and time-to-expiration $T$, the final payoff for the call option is $C_T(S,K,T)=(S_T-K)^+$ and for the put option is $P_T(S,K,T)=(K-S_T)^+$. The exact solution of the European option at time $t$ has the form of an infinite sum of Black-Scholes-like terms:}

\begin{equation} \label{eq6.1.6}
Z_t^J(S,K,T) = \sum_{n=0}^{\infty} \frac{e^{-\lambda (T-t)}(\lambda (T-t))^n}{n!} 
\int F_n (d\zeta) \times Z_t^{BS}(S e^{\zeta} e^{\mu_J (T-t)}, K, r, \sigma, T),    
\end{equation}

{\rm where $F_n$ is the distribution of the sum of $n$ independent jumps, the superscript $J$ stands for Jump-diffusion model and $Z^{BS}(\cdot)$ denotes the standard Black-Scholes formula, for example,}

\begin{equation*}
C_t^{BS}(S,K,r,\sigma,T) = S N(d_+) - e^{-r(T-t)} K N(d_-)
\end{equation*}

{\rm for European call option and}

\begin{equation*}
P_t^{BS}(S,K,r,\sigma,T) = e^{-r(T-t)} K N(-d_-)-S N(-d_+)
\end{equation*}

{\rm for European put options, with} 
\begin{equation*}
d_{\pm} = \frac{1}{\sigma \sqrt{T-t}} \big[\ln \frac{S}{K} + (r \pm \frac{1}{2} \sigma^2) (T-t) \big].   
\end{equation*}

{\rm The pricing formula of European options with jumps can be seen as an adjusted sum of several European options with different numbers $n$ of jumps.}
\end{theorem}

\begin{proof}
The value of the contingent claim, European-style option, at time $t$ is the discounted value of its expectation under the risk-neutral measure $\Q$:

\begin{align*}
Z_t^J(S,K,T) &= e^{-r(T-t)} \EQ[Z_T(S_T,K,T)| \mathcal{F}_t ] \\
&=e^{-r(T-t)} \EQ[Z_T(S_t e^{(r - \frac{\sigma^2}{2} - \lambda \zeta)(T-t) + \sigma (W_T - W_t) + (J_T-J_t)},K,T)| \mathcal{F}_t ] \\
&=e^{-r(T-t)} \EQ[Z_T(S_t e^{(r - \frac{\sigma^2}{2} - \lambda \zeta)(T-t) + \sigma W_{T-t} + J_{T-t}},K,T)| \mathcal{F}_t ] \\
&=e^{-r(T-t)} \EQ[Z_T(S_t e^{(r - \frac{\sigma^2}{2} - \lambda \zeta)(T-t) + \sigma W_{T-t} + J_{T-t}},K,T) ].
\end{align*}

By conditioning on the number of jumps $n$, we can get

\begin{align*}
Z_t^J(S,K,T) 
&=e^{-r(T-t)} \sum_{n=0}^{\infty} \Q(N_{T-t}=n) \EQ[Z_T(S_t e^{(r - \frac{\sigma^2}{2} - \lambda \zeta)(T-t) + \sigma W_{T-t} + J_{T-t}^n},K,T) ] \\
&=e^{-r(T-t)} \sum_{n=0}^{\infty} \frac{e^{-\lambda (T-t)}(\lambda (T-t))^n}{n!} \EQ[Z_T(S_t e^{(r - \frac{\sigma^2}{2} - \lambda \zeta)(T-t) + \sigma W_{T-t} + J_{T-t}^n},K,T) ],
\end{align*}

where

\begin{equation*}
J_{T-t}^n = \sum_{i=1}^{n} Y_i =\sum_{i=1}^{N_{T-t}=n} Y_i.
\end{equation*}

Notice that $\Q(N_{T-t}=n)$ is the probability of occurring $n$ jumps between $t$ and $T$, which follows a Poisson distribution.
Then, we can use jump compensator $\mu_J$ to represent the jump size inside and the explicit solution can be obtained.
\end{proof}

Moreover, according to Merton \cite{Merton1976}, we can simplify the pricing formula of European call options Eq. \ref{eq6.1.6} for normally distributed jumps as follows.

\begin{proposition} \label{normal}
If jumps are normally distributed with $Y_i \sim N(\alpha, \delta^2)$, the pricing formula for European call option at initial $C_0^n(S,K,T)$ for underlying asset with jumps can be simplified as

\begin{equation} \label{eq6.1.7}
C_0^n(S,K,T) = \sum_{n=0}^{\infty} \frac{e^{-\lambda' T}(\lambda' T)^n}{n!} 
C_0^{BS}(S, K, r_n, \sigma_n, T),    
\end{equation}

where
\begin{align*}
\lambda' &= \lambda e^{\alpha + \frac{\delta^2}{2}} \\
\sigma_n^2 T &= \sigma^2 T + n \delta^2 \\
r_n T &= (r+\mu_J)T + n(\alpha + \frac{\delta^2}{2}) 
= \big(r- \lambda(e^{\alpha + \frac{\delta^2}{2}}-1) \big)T + n(\alpha + \frac{\delta^2}{2}). 
\end{align*}

Each term of European call option $C^{BS}(S, K, r_n, \sigma_n, T)$, which is given by standard Black-Scholes formula, in Eq. \ref{eq6.1.7} is the value of the option conditional on there being exactly $n$ jumps during its life.
\end{proposition}

\begin{proof}
By assuming jumps are independent and identically normally distributed with $Y_i \sim N(\alpha, \delta^2)$, we have 
\begin{equation*}
\zeta = e^{\alpha+\frac{\delta^2}{2}}-1    
\end{equation*}
and the sum of $n$ independent jumps satisfies normal distribution that $\sum_{i=1}^n Y_i \sim N(n \alpha, n \delta^2)$. Thus

\begin{equation*}
C_T^n(S,K,T) 
=e^{-r(T-t)} \sum_{n=0}^{\infty} \frac{e^{-\lambda T}(\lambda T)^n}{n!} \EQ[C_T^n(S e^{(r - \frac{\sigma^2}{2} - \lambda(e^{\alpha + \frac{\delta^2}{2}}-1))T + \sigma W_T + J_T^n},K,T) ].    
\end{equation*} 
 
Plugging the distribution of the sum of jumps $F_n$ into the valuation formula for the European call option, Eq. \ref{eq6.1.6}, and using the change of variable for integration of probability density function of jumps sum, the term in the exponential with $N_{T-t}=n$ follows

\begin{align*}
&(r - \frac{\sigma^2}{2} - \lambda(e^{\alpha + \frac{\delta^2}{2}}-1))T + \sigma W_{T-t} + J_T^n \\
\sim & N((r - \frac{\sigma^2}{2} - \lambda(e^{\alpha + \frac{\delta^2}{2}}-1))T + n \alpha, \sigma^2 T + n \delta^2)
\end{align*}

and it can be rewritten the volatility as $\sigma_n^2 = \sigma^2 + \frac{n \delta^2}{T}$ so that its distribution stays the same 

\begin{align*}
&(r - \frac{\sigma^2}{2} - \lambda(e^{\alpha + \frac{\delta^2}{2}}-1))T + n \alpha + \sigma_n W_T \\
\sim & N((r - \frac{\sigma^2}{2} - \lambda(e^{\alpha + \frac{\delta^2}{2}}-1))T + n \alpha, \sigma^2 T + n \delta^2).
\end{align*}

By adding and subtracting $\frac{n \delta^2}{2T}$ in the exponential one gets


\begin{align*}
&(r - \frac{\sigma^2}{2} - \lambda(e^{\alpha + \frac{\delta^2}{2}}-1))T + n \alpha + \sigma_n W_T \\
& =\big(r - (\frac{\sigma^2}{2} + \frac{n \delta^2}{2T}) + \frac{n \delta^2}{2T} - \lambda(e^{\alpha + \frac{\delta^2}{2}}-1) \big)T + n \alpha + \sigma_n W_{T-t} \\
& =n (\alpha + \frac{\delta^2}{2}) - \lambda(e^{\alpha + \frac{\delta^2}{2}}-1)T+ (r-\frac{\sigma_n^2}{2})T + \sigma_n W_T .
\end{align*}

Thus we can get the pricing formula for the European call option under normally distributed jumps that

\begin{align*}
C_0^n &= e^{-rT} \sum_{n=0}^{\infty} \frac{e^{-\lambda T}(\lambda T)^n}{n!} \EQ[Z_T(S_t e^{(r_n-\frac{\sigma_n^2}{2})T + \sigma_n B_T},K,T) ] \\
& = \sum_{n=0}^{\infty} e^ {(r_n - r)T} \frac{e^{-\lambda T}(\lambda T)^n}{n!} C_0^{BS}(S,K,r_n,\sigma_n,T).
\end{align*}

Changing $\lambda$ to $\lambda'= \lambda e^{\alpha + \frac{\delta^2}{2}}$, we can put $e^ {(r_n - r)T}$ part to the probability of $n$ jumps part and Eq. \ref{eq6.1.7} can be achieved.
\end{proof}

{\bf Put-Call Parity under Jump-Diffusion Model} \label{sec6.1.4}

We need to mention that the put-call parity under the standard Black-Scholes model is not satisfied anymore. Fortunately, the value of European call options and European put options still have some other relationships. Here we have the following equation representing the new put-call parity,

\begin{theorem} \label{pc_parity} {\rm
Assuming the stock price at time $t$ to be $S_t$, the put-call parity under the jump-diffusion model is now given as follows:}

\begin{align} \label{eq6.1.9}
C^J_t(S,K,T) - P^J_t(S,K,T) &= e^{-r(T-t)} \EQ[(S_T-K)^+ - (K-S_T)^+ |\cF_t] \nonumber \\
&=e^{-r(T-t)} \EQ[S_T-K |\cF_t] \nonumber \\
&= S_t e^{-J_t} \EQ[e^{J_T}|\cF_t] - e^{-r(T-t)}K
\end{align}

{\rm where $C^J_t$ and $P^J_t$ are the values of European call options and European put options under the jump-diffusion model at time $t$. $J_t$ is a compound Poisson process defined in Eq. \ref{eq6.1.1} which represents jumps during the holding period of assets.}
\end{theorem}

\subsection{Valuation Formula with Default Risks} \label{sec6.2.1}

Lastly, we want to find the value process under the random time default model. By using the valuation formula, we can add the influence of the random time default event to the financial products under any model, such as the most general Black-Scholes model.
And of course, we can consider such default risks under the jump-diffusion model directly and it is independent of the valuation formula under the jump-diffusion model. 
As we already discussed before, we can view the random default time $\tau$, which is the first jump time of a Cox process with intensity $\gamma^P$, as a discount rate similar to the interest rate because the default time $\tau$ is independent of the financial market $(B,S,\mathbf{F})$. Then we have the following lemma according to Szimayer \cite{Szi2005}.

Unlike in the standard Black-Scholes model, here we describe any probability measure $\Q$ equivalent to $\bbP$ by its density process $L^{\psi,\phi}$. Apply Girsanov's Theorem, we have

\begin{equation} \label{rt_L}
dL_t^{\psi,\phi} =  L_{t^-}^{\psi,\phi} (\psi_t\,dW_t + (\phi_t-1)\,dM_t),
\quad \text{for} \quad 0 \leq t \leq T,
\end{equation}

and $L_0^{\psi,\phi}=1$, where $\psi$ and $\phi$ are predictable processes and $\phi$ is strictly positive with

\begin{equation*}
\int_0^T \psi_t^2\,dt < \infty, \quad
\int_0^T \phi_t \gamma_t^P\,dt < \infty \quad a.s.
\end{equation*}

Then the processes $\widetilde{W}$ and $\widetilde{M}$ given by

\begin{equation*}
\widetilde{W}_t = W_t - \int_0^t \psi_s\,ds, \quad \text{and} \quad
\widetilde{M}_t = D_t - \int_0^{t \wedge \tau} \phi_s \gamma_s^P\,ds, 
\quad \text{for} \quad 0 \leq t \leq T,
\end{equation*}

are $\Q$-martingales. Moreover, $\widetilde{W}$ is a Brownian Motion under probability measure $\Q$ and $\gamma^Q = \gamma^P (1+\phi)$ is the $\Q$-intensity of $N$.

Then the probabilistic structure of the market is briefly characterised by the following theorem:

\begin{theorem} \label{rt_Q} {\rm
Let $\Q$ be an equivalent martingale measure to real market probability measure $\bbP$, then $\Q$ is given in terms of $(\psi,\phi)$ by $d\Q = L_T^{\psi,\phi}\,d\bbP$, where $L^{\psi,\phi}$ is defined by Eq. \ref{rt_L}, and}

\begin{equation}
\psi_t = -\frac{\mu-r}{\sigma}, 
\quad \text{for} \quad 0 \leq t \leq T.   
\end{equation}

{\rm Here the market price of risk of the stock is $-\psi_t = \frac{\mu-r}{\sigma}$ is a unique value since the stock is traded and the market model $(B,S,\mathbf{F})$ is complete. However, the setting for the event risk is incomplete. The event risk is actually not reflected in any traded asset and thus, the market price of risk of the event risk $\gamma^P \phi$ is arbitrary.}
\end{theorem}

With the definition of Radon-Nikod{\'y}m derivative by Theorem \ref{rt_Q}, we can now give the standard result for a European-style contingent claim in a random time default model. This type of valuation formula is well known from the pricing of credit risk derivatives.

\begin{lemma} \label{lemma6.2.1}
Let $X_T \I_{\{\tau > T\}}$ be a non-negative $\mathcal{F}_T$-measurable random variable representing the payoff of a European option maturing at time $T$ with $\E_P(X_T) < \infty$. We have the value process $V^{Q}$ of such contingent claims on $\{\tau > t\}$ under a given equivalent martingale measure $\Q$

\begin{equation} \label{eq6.2.3}
V_t^{Q} := B_t \E_{Q} \Big\{ \frac{X_T \I_{\{\tau > T\}}}{B_T} \Big | \mathcal{G}_t \Big\}
= B_t \E_{Q} \Big\{ \frac{X_T}{B_T} \Gamma_{t,T}^{\gamma} \Big | \mathcal{F}_t \Big\}    
\end{equation}

where the final payoff of the contingent claim could be $X_T=(S_T-K)^+$ for a European call option under the random time default model and $\Gamma^{\gamma}$ is the conditional survival probability of $\tau$ with respect to the measure $\Q$:

\begin{equation} \label{eq6.2.4}
\Gamma_{t,u}^{\gamma} = Q(\tau>u | \tau>t, \mathcal{F}_T) = \exp{\Big( -\int_t^u \gamma^Q_s ds \Big)}, \quad \text{for} \quad 0 \leq t \leq u \leq T,    
\end{equation}

here $\gamma^Q$ is the $\Q$-intensity of $\tau$ where $\gamma^Q = \gamma^P (1+\phi)$, and $\phi>0$ is strictly positive determined by $\Q$.
\end{lemma}

Consider a European call option on $S$ with a strike price $K$, the event risk is the bankruptcy of the option seller. 
Thus if we want to simplify the valuation formula for European options, we can set both interest rate $r$ and intensity $\gamma^P$ of the Poisson process which corresponds to random time $\tau$ to be constants. Then by using martingale measure $\Q$, we have the value process of a European call option under random time default model with final payoff $C=(S_T-K)^+\I_{\{\tau > T\}}$ at maturity $T$ as follows

\begin{align*}
V_t^Q (C) &= \E_Q \Big\{ \exp{\Big( -\int_t^T (r+\gamma^Q)\,du \Big)} 
(S_T-K)^+ \Big | \mathcal{F}_t \Big\} \\
&= e^{-(r+\gamma^Q)(T-t)} \E_Q \Big\{(S_T-K)^+ \Big | \mathcal{F}_t \Big\} \\
&= e^{-\gamma^Q(T-t)} C_t(S,K,T)
\end{align*}

where $C(S,K,T)$ is the value function of European call options, no matter whether with consideration of jumps or not.
Similarly, for European put option with final payoff $P=(K-S_T)^+ \I_{\{\tau > T\}}$ at maturity $T$, the value process under random time default model equals to

\begin{equation*}
V_t^Q (P) = e^{-\gamma^Q(T-t)} P_t(S,K,T)    
\end{equation*}

where $P(S,K,T)$ is the value function of European put options.

\section{Hedging Strategies under Jump-Diffusion Model and Default Risks} \label{sec6.4+}

Based on the general pricing formula for the adjusted return of a standard EPS proposed by Xu et al. \cite{Xu2024}, the hedging strategy for a general standard EPS without jump and default risks can be derived naturally. Following Xu et al. \cite{Xu2024}, a generic EPS can be decomposed into two components, namely the protection leg and the fee leg, using the ideas in Lemma \ref{lem0.2.1}. These two components can then be hedged separately, leading to the following hedging result for a standard EPS with a virtually arbitrary structure of protection and fee legs.

\begin{proposition}\label{prop5.2.1}
A static hedging portfolio $H$ composed of call and put options for per unit nominal principal of standard EPS has the following payoff at time $T$
\begin{equation} \label{eq0.7.5x}
H(T)=\sum_{i=0}^{n}\frac{p_{i+1}-p_i}{S_0}\,\text{Put}_T(K^l_i,T) - \sum_{j=0}^{m}\frac{f_{j+1}-f_j}{S_0}\,\text{Call}_T(K^g_j,T)
\end{equation}
where $K^l_i=S_0(1+l_i)$ and $K^g_j=S_0(1+g_j)$ for every $i=0,1,\dots, n$ and $j=0,1,\dots ,m$.
\end{proposition}

Given the hedging portfolio $H$, let us assume that the provider of swap receives contract premium $c$ per one unit of the nominal value, which can be seen as the hedging cost for the provider and is applied to the nominal principal $N_p$ at the beginning of the contract. According to our convention, the provider receives the cash flow $N_p \pp (R_T)$ from the buyer at the maturity $T$ of the swap and the continuously compounded annualized risk-free rate is $r$ during the lifetime $[0,T]$ of the EPS. Then we have the definition of EPS provider's cash flow and fair initial premiums of the standard EPS, which have been defined by Xu et al. \cite{Xu2024}. 

\begin{definition} \label{EPS CF} {\rm
For a standard EPS introduced by Xu et al. \cite{Xu2024}, the provider's \textit{hedged cash flow} for one unit of the nominal principal $N_p$, assessed at maturity $T$ and denoted by $CF_T(c,H)$, equals
\begin{equation} \label{eq0.7.1}
CF_T(c,H)=(c-H(0))e^{rT}+H(T)+\pp (R_T)
\end{equation}
where $\pp (R_T)$ is given in Lemma \ref{lem0.2.1}. The hedging strategy presented in Proposition \ref{prop5.2.1} is \textit{a static hedging strategy} that we can find a premium $\widehat{c}\in \mathbb{R}$ such the equality $CF_T(\widehat{c},H)=0$ holds almost surely. And $\widehat{c}$ is called the \textit{fair premium} for a standard EPS per one unit of the nominal principal.}
\end{definition}

With the static hedging strategy for the standard EPS products, we now want to find the hedging strategies under consideration of jumps and default risks. We will start with jumps only, without default risks.

\subsection{Hedging Strategies without Default Risk} \label{sec6.4.1}

We first consider hedging strategies without default risk, under the consideration of jumps. Actually, we can see that if there is a jump during the lifetime of the underlying asset, both values of reference portfolios and European options will be changed. With the possibility of a jump, the values of call options decrease and the values of put options increase. 
Under an EPS, we consider the final return of the reference portfolio that should be hedged by the swap provider, thus the swap provider can still use the same hedging strategies discussed before to cover the return of the reference portfolio even when there are jumps. And the changeable values of options can represent the influence of jumps in the hedging strategies.

We can see that the hedging strategies which are displayed by Proposition \ref{prop5.2.1} can still fully hedge standard EPS with jumps. However, we need to mention that there is an important assumption here that no default risk in the market. In the meanwhile, we should now consider the pricing of European options with possible jumps, which will be discussed in Section \ref{sec6.3} as part of numerical studies for hedging costs. If we assume there is no default risk, then the assumptions of exactly one jump or no more than one jump will not impact the structure of hedging strategies. No matter there is a jump or not, the European options can cover the outflows and inflows of a standard EPS, but the values of European options will be different under different models and it will influence the fair initial premiums of EPS products.
We will do numerical studies later.

\subsection{Only Counter Party with Default Risk} \label{sec6.4.3}

Now we will discuss the influence of default risks on the hedging strategies of standard EPS. Firstly, it is important to note that EPS becomes useless if we assume all financial institutions (including the EPS provider and its counter party) will default when a jump occurs. In this case, the swap buyers (variable annuity holders) cannot get any payoff from EPS providers and their variable annuity loss according to jump cannot be covered by this protection (standard EPS), if a jump happens. The main reason that investors buy these products is to protect themselves from large losses of variable annuity which cannot be reached if all financial institutions will default when a financial crisis happens, it makes these products meaningless.

In this case, we want to have a look at a more important situation, only the counter-party (put option sellers, excluding EPS providers) has a possibility to default. 
When we assume only the counter-party will possibly default and EPS providers will not default, it means that the EPS will not be affected by jumps and EPS buyers can still get the cash flows because of the price slide of the reference portfolio. Thus a situation in that EPS providers will not default is better for investors, and we want to see the potential loss of EPS providers when a jump occurs.

Before we evaluate the influence of jumps on hedging strategies, we want to discuss their impact on options. With a negative jump, call option buyers and sellers will not default as the final payoff is zero almost surely. Only for put options, the option sellers may default when a jump occurs, because negative jumps in assets’ prices cause large debt payoffs that they cannot afford. Thus we can conclude that call options will not be influenced by jumps and the potential gains can be fully hedged by hedging strategies introduced before. So we only need to discuss the hedging strategies with put options.

We still start with the hedging strategies structured in normal cases without jumps, which has been discussed by Xu et al. \cite{Xu2024}. From the hedging strategies built in Proposition \ref{prop5.2.1}, there are both long positions and short positions of put options for EPS providers to hedge the cash flows of EPS. 
If jumps happen, the EPS provider cannot get any payoff from the long positions of put options because the counter-party (put option sellers) defaults. In the meanwhile, the EPS providers will still pay the debt obligations that come from their short positions of put options. Moreover, they will comply with the rules of EPS that make cash outflows to EPS buyers for price decreases in the reference portfolio and gain cash inflows from EPS buyers for price increases in the reference portfolio. The EPS providers will face a huge debt obligation when there is a negative jump.

Now we want to see some typical examples of swap providers' cash flow after hedge when a financial crisis happens, buffer EPS and floor EPS. By assuming only the counter-party (European option sellers) has the possibility of default and the EPS provider will not default, there is no influence on the EPS contract between EPS buyers and EPS providers. We only need to find the EPS providers' final cash flows when there is a jump in the holding period of the EPS products.  

\begin{example} \label{ex6.4.1}
{\rm Consider a standard buffer EPS introduced in Section \ref{sec6.0}, we use a static hedging strategy to hedge the potential gain or loss from this EPS, which has the hedging cost as follows,}

\begin{equation*}
H(0) =\textcolor{blue}{\frac{p_2}{S_0}\,\text{Put}_0(K^l_1,T)}
- \frac{f_2}{S_0}\,\text{Call}_0(K^g_1,T).   
\end{equation*}

{\rm If a jump happens, then the put option $\text{Put}(K^l_1,T)$ (the blue one), that the EPS provider bought in order to hedge their potential loss, worth nothing because of the third party default. Fortunately, the call options, no matter the EPS providers' long or short will not be defaulted because they will not be implemented. Thus the static hedging strategy before has the final payoff }

\begin{equation*}
H^D(T) = - \frac{f_2}{S_0}\,\text{Call}_T(K^g_1,T)   
\end{equation*}

{\rm when there is a financial crisis. Assuming a jump happens, the EPS provider's cash flow at maturity $T$ after hedge now can be represented as}

\begin{align} \label{eq6.4.1}
CF_T^D(c,H)&=(c-H(0))e^{rT}+H^D(T)+ \pp_B (R_T)\\
&=(c-H(0))e^{rT}+H^D(T) - p_2 (l_1-R_T)^+ + f_2 (R_T-g_1)^+ \nonumber \\
&=(c-H(0))e^{rT} - p_2 (l_1-R_T)^+ \nonumber
\end{align}

{\rm where $CF^D$ and $H^D$ represent the cash flows and hedging portfolios of EPS products in which the counter party defaults at the maturity respectively.} 
\end{example}

\begin{example} \label{ex6.4.2}
{\rm Consider another example of standard EPS, floor EPS, introduced in Section \ref{sec6.0}, EPS providers have the following static hedging portfolio,}

\begin{equation*}
H(0) = \textcolor{red}{-\frac{p_1}{S_0}\,\text{Put}_0(K^l_1,T)}
\textcolor{blue}{+ \frac{p_1}{S_0}\,\text{Put}_0(S_0,T)}
- \frac{f_2}{S_0}\,\text{Call}_0(K^g_1,T) 
\end{equation*}

{\rm The call options do not have any default probability that we only need to focus on put options. 
For the first short positions of the put option (the red one), $\text{Put}(K^l_1,T)$, we assume the EPS providers will not default and they will still pay the debt obligation even when the financial crisis comes. 
However, for the second long positions of put options (the blue one), $\text{Put}(S_0,T)$, the EPS provider will not get any payoff when a jump occurs and third-party default. In this case, we have the value of hedging portfolio at maturity $T$ when the jump happens as follows,}

\begin{equation*}
H^D(T) = \textcolor{red}{- \frac{p_1}{S_0} \, \text{Put}_T(K^l_1,T)}
- \frac{f_2}{S_0} \, \text{Call}_T(K^g_1,T).
\end{equation*}

{\rm And the swap provider's final cash flow at time $T$ after hedge has the following representation, assuming jump occurs,}

\begin{align} \label{eq6.4.2}
CF_T^D(c,H) &=(c-H(0))e^{rT}+H^D(T)+\pp_F (R_T) \\
&= (c-H(0))e^{rT}+H^D(T)-p_1 (-R_T)^++p_1(l_1-R_T)^++f_2(R_T-g_1)^+ \nonumber \\
&= (c-H(0))e^{rT}-p_1 (-R_T)^+ \nonumber
\end{align}

{\rm where $CF^D$ and $H^D$ represent the cash flows and hedging portfolios of EPS products with the default event from the counter party respectively.}
\end{example}

Assuming the initial premium of EPS $c$ to be fair under the normal case that $c=H^D(0)$, we can simplify the final cash flow for buffer EPS, Eq. \ref{eq6.4.1}, and for floor EPS, Eq. \ref{eq6.4.2}, to be only one term left. This term is actually the EPS providers' potential loss after hedging because of a third party's default when there is a financial crisis.
The following figures, Figure \ref{tp_default}, show the swap provider's cash flow after hedge for our two examples - buffer EPS and floor EPS, considering there is a default event from the counter party in the holding period of the reference portfolio. As a reminder, we assume the initial premium of EPS $c$ equals the cost of hedging strategies $H(0)$.

\begin{figure} [h!]
    \centering
    \subfloat[\centering buffer EPS]
    {{\includegraphics[width=7cm,height=5.5cm]{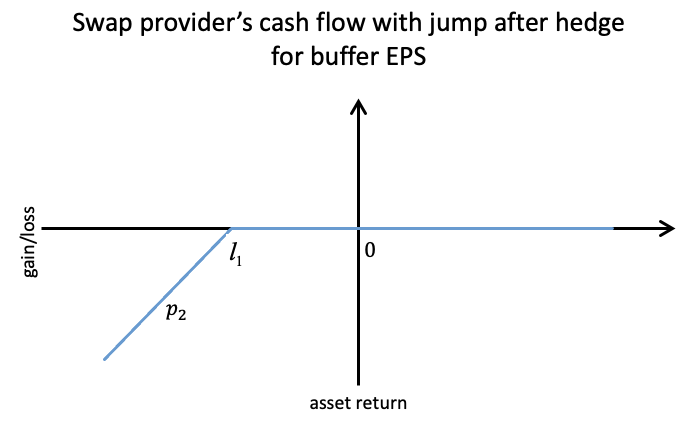} }}%
    \qquad
    \subfloat[\centering floor $\&$ floor-cap EPS]
    {{\includegraphics[width=7cm,height=5.5cm]{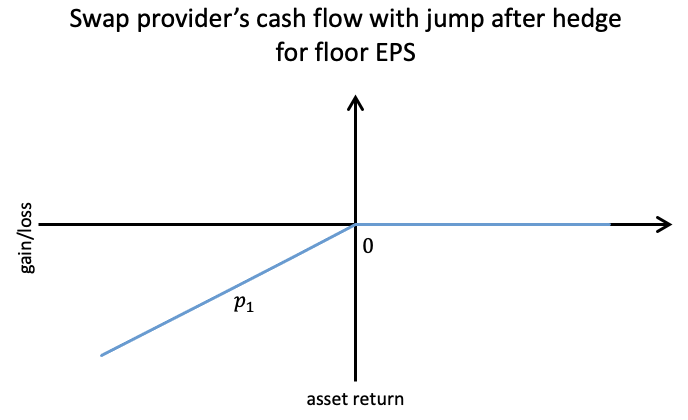} }}%
    \caption{Examples of Swap provider's cash flow after hedge when only third party default}%
    \label{tp_default}%
\end{figure}

From these examples, we can see that EPS providers will have an additional loss more than the EPS cash flow itself when there is a default event, because they also face the loss from put option sellers defaulting, together with EPS debt obligation.
The graph of the swap provider's cash flow with default events after hedge is an upward strict line with all negative values then a horizontal line with value zero, where the vertical axis represents gain or loss and the horizontal axis represents asset return $R_T$. The turning point of this graph is the larger point of the first loss interval $[l_{i+1},l_i)$ that participation rate $p_{i+1}$ is greater than zero. We set this turning point to be $\widehat{l}=l_i$ with participate rate $\widehat{p}=p_{i+1}$.
Moreover, the slope of the upward strict line is equal to the first non-zero participation rate for a loss of the underlying asset $\widehat{p}$. 
Thus we have the following proposition for this situation.

\begin{proposition} \label{prop_td}
Assuming only the counter-party has default probability, we can represent the swap provider's after-hedged cash flow at maturity $T$ as follows, together with hedging portfolios structured in Proposition \ref{prop5.2.1},

\begin{equation} \label{eq6.4.3}
CF_T^D(c,H) = (c-H(0))e^{rT} - \widehat{p} (\widehat{l}-R_T)^+ 
\end{equation}

where $CF^D$ means the EPS provider's cash flow in which the counter party defaults in the holding period of reference portfolio and $\widehat{p}$ is the first positive participation rate for loss interval. 
\end{proposition}

The following figure, Figure \ref{EPS_td} presents the swap provider's cash flow after hedge for a general standard EPS, considering only the third-party default.

\begin{figure} [h!]
    \centering
    \includegraphics[width=12cm, height=8cm]{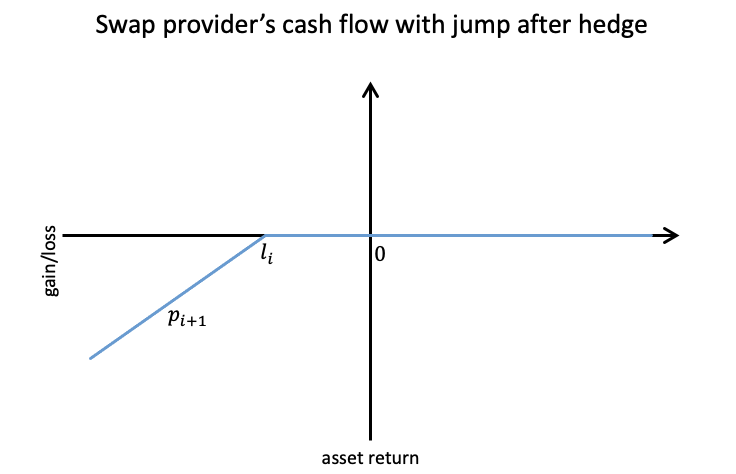}
    \caption{EPS provider's general cash flow after hedge when only third party default}
    \label{EPS_td}
\end{figure}

It can be seen that the EPS provider needs to cover a huge loss if the counter parties default. Fortunately, we need to mention that the loss is limited, the maximum loss of the portfolio is $-100\%$ that the value of the underlying asset drops to zero. 
Moreover, the cash flow with consideration of defaults is a random number, so that a "fair" initial premium $c$, which makes the after-hedge cash flow equal to zero $CF_T^D(c,H)=0$, cannot be predetermined before the maturity $T$. In this case, we need to apply the expectation to the cash flow equation in order to find a proper initial premium, called as default adjusted initial premium. 
We will try to find the risk-neutral pricing of EPS products under this assumption in the next section.

\subsection{Default Adjusted Initial Premium when Only Counter Party with Default Risk} \label{sec6.4.4}

After discussing the EPS provider's potential cash flows when only the counter party has default probabilities, we now need to find a default adjusted initial premium for EPS under the jump model with consideration of default risks. As we already mentioned, there are two parties we need to consider for default risk, the EPS provider and all other counter parties. Assuming only the EPS provider has default risks, we can conclude that the original hedging strategies without consideration of jumps and credit risks are actually super-hedging strategies for the EPS provider. In the meanwhile, if all financial institutions default when the jump happens, the EPS provider will default to the EPS buyer and the EPS products become useless for investors as there is no more loss protection. 
Thus we only need to focus on the most important situation the third party with default risk, where the hedging strategy in normal cases cannot fully hedge the potential losses of the EPS provider. Moreover, there are no derivatives that can fully hedge the influence of defaults and we can only try to find default adjusted initial premiums under general hedging strategies for the EPS provider. 

Fortunately, we can use expected loss after hedge (here we call it default adjustment in order to distinguish losses we discussed before) as a simple indicator of risk which only considers a single default probability and loss severity. We can now make a definition for the default adjusted initial premium by using the default adjustment.

\begin{definition} \label{def_cd}
{\rm Considering the static hedging strategy $H$ that can fully hedge the standard EPS product under the general situation without defaulting, the default adjusted initial premium $c^D$ should be defined in the following formula, under the assumption that only the counter party of the EPS provider has default probabilities,}

\begin{equation} \label{cd}
c^D := H(0) + \frac{B_0}{B_T} DA
\end{equation}

{\rm where $H(0)$ is the initial hedging cost, $DA$ is the EPS provider's default adjustment under the default risks' assumption and it should be discounted from maturity to initial time. The default adjustment can be calculated as follows:}

\begin{equation*}
\text{Default Adjustment} = \text{Loss Severity} \times \text{Default Probability}.
\end{equation*}
\end{definition}

As a remark, the loss severity in the default adjustment function can be rewritten as the expected potential loss conditional on the probability of default that 

\begin{equation*}
DA = \E(DA|P^D) \times P^D.
\end{equation*}

According to its definition, we can find the default adjusted initial premium for EPS to discuss both the normal case and under the default model. Firstly, the hedging portfolio has the same structure as discussed in Proposition \ref{prop5.2.1}. In the meanwhile, the initial premium $c$ with consideration of default risks paid by EPS buyers should be greater than the fair initial premium in the general case equals $H(0)$, in order to cover the EPS providers' potential loss under the jump model. Compared with fair initial premium $\widehat{c}$, this new default adjusted initial premium $c^D$ contains the conditional default adjustment of the EPS provider given by loss severity under default events multiplied by the default probability. In the next step, we want to find the probability of default under the random time default model first when considering an independent default event. Furthermore, we will try to find the default adjustment under both the vanilla Black-Scholes model and the jump-diffusion model.

{\bf Independent random time default probability}

\begin{proposition} \label{default_prob}
Consider the random time default model, which is introduced in Section \ref{sec6.2}, random time $\tau$ is used to present the time of jump happens. Assuming the density $\gamma^Q$ for indicator process $N$ under probability measure $\Q$ to be constant and stopping time (the first jump time) $\tau$ satisfies exponential distribution, we can use Proposition \ref{prop6.2.1} to get the probability of default happening under random time default model as following,

\begin{equation*}
P^D:=P(\tau < T) = 1- P(\tau > T) = 1- \EQ \Big\{ \exp{\Big(-\int_0^T \gamma_u^Q du \Big)} \Big\} = 1- e^{-\gamma^Q T}.    
\end{equation*}

\end{proposition}

{\bf Independent random time default adjustment}

Then we want to find the default adjustment conditional on the default probability as loss severity in order to get the whole default adjustment.
Before we discuss the conditional default adjustment, from the after static hedging cash flow of the EPS provider displayed in Section \ref{sec6.4.3} and Figure \ref{EPS_td}, we can see that if a default event occurs, there is a loss after hedging starting from the larger point $\widehat{l}$ of the first loss interval with a positive participation rate $\widehat{p}$. Fortunately, the value of the reference portfolio will only tend to zero, not be negative, so that this potential loss after hedging is bounded by the initial value of the reference portfolio $S_0$. Thus the largest potential loss after hedging that the EPS provider needs to cover for the EPS buyer is $(1+\widehat{l})$, multiplied by the EPS provider's participation rate, the maximum loss after hedging is equal to $\widehat{p}(1+\widehat{l})$ for per unit of nominal principal $N_p$. 
Moreover, we can use the distributions of the value of the underlying asset $S_T$ at maturity to find the conditional default adjustment.

\begin{proposition} \label{exp_loss_ind}
Consider the dynamics of the underlying asset $S$ satisfies the Black-Scholes model with random time default event according to Definition \ref{random_t}, the conditional default adjustment $\E(DA|P^D)$ for per unit of the nominal principal $N_p$ at the maturity $T$ will be given as

\begin{equation*}
\E(DA|P^D) = \widehat{p} e^{rT} N \bigg( \frac{\ln(1+\widehat{l}) - (r+ \frac{1}{2} \sigma^2)T}{\sigma \sqrt{T}} \bigg)    
\end{equation*}

where $N$ represents the cumulative density function of the standard normal distribution and $\widehat{l}$ is the larger point of the first loss interval with positive participation rate $\widehat{p}$. Then the EPS provider's default adjustment under the random time default model is equal to

\begin{equation} \label{el_eq1}
DA = (1-e^{-\gamma^Q T}) \widehat{p} e^{rT} N \bigg( \frac{\ln(1+\widehat{l}) - (r+ \frac{1}{2} \sigma^2)T}{\sigma \sqrt{T}} \bigg).   
\end{equation}

For the jump-diffusion model with the underlying asset $S$ satisfying Eq. \ref{eq6.1.1}, the default adjustment conditional on $n$ jumps in the holding period of the underlying asset with jump size $Y_i \sim N(\alpha, \delta^2)$ can be calculated as

\begin{equation*}
\E(DA|N_T=n) 
=\widehat{p} e^{r_n T} N \bigg( \frac{\ln(1+\widehat{l}) - (r_n + \frac{1}{2} \sigma_n^2)T}{\sigma_n \sqrt{T}} \bigg),
\end{equation*}

at the maturity, where $r_n$ and $\sigma_n$ have been defined in the jump-diffusion model section. With the same algorithm we discussed before, it can be concluded that the default adjustment for per unit of the nominal principal $N_p$ will be given as follows under the assumption of independent default,

\begin{equation} \label{el_eq2}
DA = (1-e^{-\gamma^Q T}) \widehat{p} \sum_{n=0}^{\infty} \frac{e^{-\lambda T}(\lambda T)^n}{n!} e^{r_n T} N \bigg( \frac{\ln(1+\widehat{l}) - (r_n + \frac{1}{2} \sigma_n^2)T}{\sigma_n \sqrt{T}} \bigg).   
\end{equation}
\end{proposition}

\begin{proof}
Firstly, we consider the random time default model only, the conditional expected loss is actually independent of the default probability because the default event is independent of the valuation process of the underlying asset. We have the distribution of the underlying asset at the maturity $T$ that $S_T \sim S_0 N((r-\frac{1}{2}\sigma^2)T,\sigma^2 T)$. Using ${S_T/S_0 \leq 1+\widehat{l}}$, the conditional default adjustment for per unit of the nominal principal can be calculated by

\begin{align*}
\E(DA|P^D) &= \E(DA) = \widehat{p} \int_{-\infty}^{\ln(1+\widehat{l})} e^u \phi \bigg( \frac{u-(r-\frac{1}{2}\sigma^2)T}{\sigma \sqrt{T}} \bigg) \,du \\
&= \widehat{p} \int_{-\infty}^{\ln(1+\widehat{l})} \phi \bigg( \frac{u-(r+\frac{1}{2}\sigma^2)T}{\sigma \sqrt{T}} \bigg) e^{rT} \,du \\
&=\widehat{p} e^{rT} N \bigg( \frac{\ln(1+\widehat{l}) - (r+ \frac{1}{2} \sigma^2)T}{\sigma \sqrt{T}} \bigg),
\end{align*}

where $\phi(.)$ is the probability density function for standard normal distribution. Multiplying the conditional expected loss with the default probability $P^D=1-e^{-\gamma^Q T}$ that we have discussed in Proposition \ref{default_prob}, we will be given the expected loss.

Similarly, if we consider an independent default event under the jump-diffusion model, we now have the distribution of the underlying asset with $n$ jumps in the holding period as $S_T \sim S_0 N((r - \frac{\sigma^2}{2} - \lambda(e^{\alpha + \frac{\delta^2}{2}}-1))T + n \alpha, \sigma^2 T + n \delta^2)= S_0 N((r_n-\frac{1}{2}\sigma^2)T,\sigma_n^2 T)$, which has been identified in Section \ref{sec6.1.2}. Still using the assumption that the underlying asset and default event are independent, we have

\begin{equation*}
DA = \E(DA|P^D) \times P^D = P^D \times \sum_{n=0}^{\infty} 
\Q(N_T=n) \times \E(DA|N_T=n),  
\end{equation*}

then we will have the default adjustment.
\end{proof}

{\bf Default-adjusted initial premium}

In conclusion, with the default adjustment for the independent random time default events, we have the following proposition for the default adjusted initial premium.

\begin{proposition} \label{default_adj_c}
Under a general hedging strategy without consideration of default, the EPS provider uses a hedging portfolio built in Proposition \ref{prop5.2.1} with an initial value at time $0$ as follows,

\begin{equation*}
H(0)=\sum_{i=0}^{n}\frac{p_{i+1}-p_i}{S_0}\,\text{Put}_0(K^l_i,T) - \sum_{j=0}^{m}\frac{f_{j+1}-f_j}{S_0}\,\text{Call}_0(K^g_j,T)
\end{equation*}

where $K^l_i=S_0(1+l_i)$ and $K^g_j=S_0(1+g_j)$ for every $i=0,1,\dots, n$ and $j=0,1,\dots ,m$. 
Assuming only the counter party will default, the default adjusted initial premium $c^D$, which is used for justifying the influence of default risks according to Definition \ref{def_cd}, received by the EPS provider per unit of the nominal principal $N_p$ in this hedging strategy is given as

\begin{equation*}
c^D = H(0) +  e^{-rT} DA   
\end{equation*}

where $H(0)$ is given before, the same as in the general situation, and $DA$ is the default adjustment that has been given in Proposition \ref{exp_loss_ind}.
\end{proposition}

Here we have another simpler pricing formula for default adjusted initial premium $c^D$. We can use the largest default adjustment after hedging $\widehat{p}(1+\widehat{l})$ (which corresponds to the value of the underlying asset drops to zero at the maturity) and the whole probability of the jump occurs as an additional compensator for this default adjusted initial premium, to cover the after-hedged potential loss of the EPS provider. From Proposition \ref{default_prob}, the default probability when considering independent default event can be given as $1-e^{-\gamma^Q T}$. Then multiplying the largest default adjustment and default probability, we have $c^{SD} = H(0) + e^{-rT} \widehat{p} (1+\widehat{l}) (1-e^{-\gamma^Q T})$, we use $c^{SD}$ represent this "super-hedging" default adjusted initial premium. In the real market, the actual loss after hedging in a financial crisis should be less or equal to the largest default adjustment we included in the initial premium $c^{SD}$ under the jump model. The EPS provider could fully cover the expected potential loss after hedging with this "super-hedging" default-adjusted initial premium, and there are still some remaining premiums if there is no default or small losses of the underlying asset.

\section{Numerical Studies for Jumps and Credit Risks} \label{sec6.5}

We have already introduced the jump-diffusion model from Merton \cite{Merton1976} in Section \ref{sec6.1} and the random time default event from Szimayer \cite{Szi2005} in Section \ref{sec6.2}. In the meanwhile, we discussed the EPS provider's after-hedge cash flows for standard EPS under the jump-diffusion model with and without consideration of the default event. 

As jumps really happen in the real market with huge influences, now we want to do numerical studies to see the hedging costs for standard EPS under jump models and compare these with the normal case without jumps introduced by Xu et al. \cite{Xu2024}.
We will consider different cases in the numerical study, starting with independent default events under the random time default model, different jump models, and different types of default events under the jump-diffusion model.
Before we simulate hedging strategies, first we want to see the changes in the valuation of European options from Vanilla European options to European options under jump models, which will be displayed in Section \ref{sec6.5.1}. It can be seen as a background to find hedging costs because our hedging portfolios contain several European options. 

Furthermore, we will discuss the influence of default risks by using the random time default model and the changes in the initial premiums according to jumps and default events under the jump-diffusion model.
For the analysis of default risks, we need to focus on the situation that only the counter party of the EPS provider has default possibilities that the EPS provider should be assumed to have no default risks to the investors. We need to mention that we cannot find a static hedging strategy that fully hedges the EPS provider's risks when jumps may happen and financial institutions have default risks. Fortunately, with some developments to the hedging strategies in the general case without jumps, we find default adjusted initial premiums for the situation where only the counter party of the EPS provider has default risk under jump models.

\subsection{Valuation of European Options with Single Jump} \label{sec6.3}

Before we try to analyse hedging strategies after considering jumps in the reference portfolio, we will start with the simplest case under the jump model - single jump. So in this section, we want to find the pricing formula for European options with a single jump. 
We can consider this situation in both two jump models discussed before - the jump-diffusion model in Section \ref{sec6.1} and the random time default model in Section \ref{sec6.2}. Moreover, we need to mention that in this single jump model, there are still two cases we want to discuss - exactly one jump and no more than one jump in the holding period. 

When considering a single jump, we make a strong assumption that the jumps will not happen anymore after the first jump, or actually, we don't study the later jumps. With this strong assumption, we can see that no matter for exactly one jump situation or no more than one jump situation, this single jump model only researches several particular situations in the whole jump model, without considering jumps after the first possible jump. However, it is still meaningful for us to study the structures of EPS and its hedging strategies so that we can have a look at the influence of certain jumps.  

In this section, our next step is to find the valuation formula of European options when there is exactly one jump in the holding period of the reference portfolio.

{\bf Valuation Formula for European Options with Single Jump} \label{sec6.3.1}

Firstly, we want to find the valuation formula of European options when there is exactly one jump. We will use the jump-diffusion model, conditioning on the event of exactly one jump here, and find some explicit pricing formulas for the European options.
According to Proposition \ref{normal}, we can conclude that if there is only one jump, the European options prices in the jump-diffusion model are given as follows:

\begin{proposition} \label{normal_one}
If there is exactly one jump in the holding period of the reference portfolio which is normally distributed with $Y \sim N(\alpha, \delta^2)$ under the jump-diffusion model, the pricing formula for the European call option at the initial $C_0^{n1}(S,K,T)$ for underlying asset $S$ with jump can be represented as

\begin{equation} \label{eq6.3.1}
C_0^{n1}(S,K,T) = e^{-\lambda T(e^{\alpha + \frac{\delta^2}{2}}-1)} e^{\alpha + \frac{\delta^2}{2}} 
C_0^{BS}(S, K, r_1, \sigma_1, T),    
\end{equation}

where 
\begin{align*}
\sigma_1^2 T &= \sigma^2 T + \delta^2 \\
r_1 T &= (r+\mu_J)T + \alpha + \frac{\delta^2}{2}. 
\end{align*}

Similarly, the pricing formula for European put option $P^{n1}(S,K,T)$ for underlying asset $S$ with exactly one jump in the holding period can be represented as

\begin{equation} \label{eq6.3.2}
P_0^{n1}(S,K,T) = e^{-\lambda T(e^{\alpha + \frac{\delta^2}{2}}-1)} e^{\alpha + \frac{\delta^2}{2}}
P_0^{BS}(S, K, r_1, \sigma_1, T).    
\end{equation}
\end{proposition}

\begin{proof}
We first have the explicit pricing formula of European options with particular $n$ jumps from the sum part of Proposition \ref{normal}, together with its weight under the jump-diffusion model. We need to mention that this weight and the pricing formula of European options under jumps are multiplied together which forms each tern of the sum. The weight is actually the probability of exactly $n$ jumps happening in the holding period, with a sum equal to $1$, which is presented as the first part of the sum in Theorem \ref{jump_option}. 

Here, we make a strong assumption that there is exactly one jump in the holding period. We can take the pricing formula from Eq. \ref{eq6.1.7} by choosing $n=1$ and now there is no weight of one jump, it is actually equal to $1$ because of the strong assumption. Thus we should divide the pricing formula of European options with one jump from Eq. \ref{eq6.1.7} by the relative weight before, $e^{-\lambda T}(\lambda T)$, and we can get the pricing formula of a European call option with exactly one jump in the holding period.

Using the pricing formula of European put options under the jump-diffusion model or put-call parity, we can get the valuation of European put options with exactly one jump.
\end{proof}

As a reminder, here we can use either the valuation formula of European options or put call parity discussed in Theorem \ref{pc_parity} to price the European options.

Under the random time default model, we actually assume that there is no more than one jump during the holding period at a random time $\tau$. We cannot decide whether this jump will happen or not, thus we cannot give an explicit pricing formula for European options with exactly one jump.

{\bf Valuation Formula for European Options with No More than One Jump} \label{sec6.3.2}

We also want to find the pricing formula for European options when there is no more than one jump - no jump or one jump. We need to mention that this is the most important case when we consider default risk. If there is exactly one jump, then we can say that financial institutions will be unable to pay their debts almost surely when default risks exist. However, if we cannot know whether there is a jump or not, then default risk becomes a probability that the financial institutions will default when a jump occurs and the financial institutions will not default when there is no jump.

The valuation formula of European options with no more than one jump can be regarded as the combination of two parts, options without jumps and options with one jump.
Thus we have the following propositions to present the value of European options.

\begin{proposition} \label{normal_z}
If there is no more than one jump in the reference portfolio which is normally distributed with $Y \sim N(\alpha, \delta^2)$ under the jump-diffusion model, the pricing formula for European call option $C^{n0}(S,K,T)$ for underlying asset $S$ with jump can be represented as

\begin{equation} \label{eq6.3.5}
C_0^{n0}(S,K,T) = \frac{e^{-\lambda' T}}
{e^{-\lambda T} (1+\lambda T)} C_0^{BS}(S, K, r+\mu_J, \sigma, T) 
+\frac{e^{-\lambda' T}(\lambda' T)}{e^{-\lambda T} (1+\lambda T)} 
C_0^{BS}(S, K, r_1, \sigma_1, T),    
\end{equation}

where 
\begin{align*}
\lambda' &= \lambda e^{\alpha + \frac{\delta^2}{2}} \\
\sigma_1^2 T &= \sigma^2 T + \delta^2 \\
r_1 T &= (r+\mu_J)T + \alpha + \frac{\delta^2}{2}. 
\end{align*}

Moreover, we have the pricing formula for European put option $P_0^{n0}(S,K,T)$ for underlying asset $S$ with a jump as follows

\begin{equation} \label{eq6.3.6}
P_0^{n0}(S,K,T) = \frac{e^{-\lambda' T}}
{e^{-\lambda T} (1+\lambda T)} P_0^{BS}(S, K, r+\mu_J, \sigma, T) 
+\frac{e^{-\lambda' T}(\lambda' T)}{e^{-\lambda T} (1+\lambda T)}
P_0^{BS}(S, K, r_1, \sigma_1, T).    
\end{equation}
\end{proposition}

\begin{proof}
This proof is similar to the proof in Proposition \ref{normal_one}, with some adjustments. Here we need to add two situations together, zero jumps and only one jump, the valuation formula of European options under these two situations with their weights can be found in Theorem \ref{jump_option}. 

Now we need to find the relative weights of these two situations to make the sum of weights to be $1$. We find the relative weights by dividing the sum of two weights $e^{-\lambda T} (1+\lambda T)$, $e^{-\lambda T}$ for zero jump and $e^{-\lambda T}(\lambda T)$ for one jump, into the pricing formula of European options. After division, we can get the pricing formula of the European call option with no more than one jump in the holding period as presented by Eq. \ref{eq6.3.5}.

Similarly, we can get the pricing formula for European put options with no more than one jump.
\end{proof}

Here we can use either the valuation formula of European options or put call parity discussed in Theorem \ref{pc_parity} to price the European options.  
Together with Section \ref{sec6.3.1}, we get the valuation formulas of European options both when there is exactly one jump and when no more than one jump.
After finding the pricing formulas for European options with a single jump, we can now evaluate the hedging costs for standard EPS.

\subsection{Numerical Study for European Options under Jump Model} \label{sec6.5.1}

In this section, we do a numerical study for European options, considering the Vanilla case without jumps and default risks, under the jump-diffusion model, and under the Vanilla model with random time default events. 
The following Table \ref{table:options_ud_jump} shows the pricing of European options under the Vanilla case (no jump) and under the jump model and we only consider negative jumps in the holding period. Several cases have been considered in the table, vanilla case (no jump), exactly one jump, no more than one jump, several jumps (consider maximum 20 jumps) in the jump-diffusion model, and in the random time default model. 

The parameters for options are defined as follows: interest rate $r=1.5\%$, volatility $\sigma=20\%$, the initial price of stock $S_0 = 100$, strike price $K=100$, and time to maturity $T=1 year$. Other parameters related to jump models, jump intensity $\lambda$, jump compensator $\mu_J$, and jump size are given in Table \ref{table:options_ud_jump}. In Merton's jump-diffusion model, we consider jump sizes to comply with a normal distribution that $Y_i \sim N(\alpha,\delta^2)$ and negative mean values $\alpha<0$ (which represent negative jumps). After assuming the distribution of jump size, we can now simplify jump compensator $\mu_J$ for the jump-diffusion model to be minus jump intensity $-\lambda$ times the mean value of the jump that $\mu_J = -\lambda (e^{\alpha+\frac{\delta^2}{2}}-1) \approx -\lambda \alpha$.

\begin{table}[h!]
\centering
\begin{tabular}{ccc|c|ccc|c}
\hline
\hline
\multicolumn{3}{c|}{Call options (parameters)} &
Vanilla & 
\multicolumn{3}{c|}{Jump diffusion}&
Default \\
 
\hline
$\lambda$ & $\mu_J$ & jump size 
&Call
&Call(1j) &Call($\leq 1$j) & Call
&Call
 \\
\hline
0.1& 0.02& $N(-0.2,0.1^2)$ 
&8.6728 &2.8261 &9.2014 &9.2176 &7.8475 \\
0.1& 0.04& $N(-0.4,0.1^2)$ 
&8.6728 &0.5159 &10.1139 &10.0892 &7.8475 \\
0.2& 0.02& $N(-0.1,0.04^2)$ 
&8.6728 &5.0674 &9.0538 &9.0293 &7.1007 \\
0.2& 0.04& $N(-0.2,0.1^2)$
&8.6728 &3.3070 &9.8030 &9.7669 &7.1007 \\
0.2& 0.08& $N(-0.4,0.1^2)$
&8.6728 &0.7780 &11.6704 &11.5277 &7.1007 \\
0.2& 0.08& $N(-0.4,0.15^2)$
&8.6728 &1.2416 &11.7364 &11.6288 &7.1007 \\
0.3& 0.03& $N(-0.1,0.04^2)$
&8.6728 &5.4587 &9.3196 &9.2094 &6.4250 \\
0.3& 0.06& $N(-0.2,0.1^2)$
&8.6728 &3.8491 &10.4753 &10.3201 &6.4250 \\
0.3& 0.12& $N(-0.4,0.1^2)$
&8.6728 &1.1449 &13.3333 &12.9650 &6.4250 \\
0.5& 0.1& $N(-0.2,0.1^2)$
&8.6728 &5.1332 &12.0284 &11.4367 &5.2603 \\
0.5& 0.2& $N(-0.4,0.15^2)$
&8.6728 &3.1386 &17.2000 &16.1328 &5.2603 \\

\hline
\hline

\multicolumn{3}{c|}{Put options (parameters)} &
Vanilla & 
\multicolumn{3}{c|}{Jump diffusion}&
Default \\
 
\hline
$\lambda$ & $\mu_J$ & jump size 
&Put
&Put(1j) &Put($\leq 1$j) & Put
&Put
\\
\hline
0.1& 0.02& $N(-0.2,0.1^2)$
&7.1840 &17.4252 &7.3399 &7.8559 &6.5004 \\
0.1& 0.04& $N(-0.4,0.1^2)$
&7.1840 &28.8330 &7.6499 &9.4135 &6.5004 \\
0.2& 0.02& $N(-0.1,0.04^2)$ 
&7.1840 &11.1902 &7.1528 &7.5982 &5.8818 \\
0.2& 0.04& $N(-0.2,0.1^2)$
&7.1840 &16.1716 &7.3178 &8.5375 &5.8818 \\
0.2& 0.08& $N(-0.4,0.1^2)$
&7.1840 &26.0458 &7.8019 &11.7282 &5.8818 \\
0.2& 0.08& $N(-0.4,0.15^2)$
&7.1840 &26.1697 &7.7962 &11.6317 &5.8818 \\
0.3& 0.03& $N(-0.1,0.04^2)$
&7.1840 &10.6502 &7.0349 &7.8082 &5.3221 \\
0.3& 0.06& $N(-0.2,0.1^2)$
&7.1840 &14.9543 &7.1668 &9.2283 &5.3221 \\
0.3& 0.12& $N(-0.4,0.1^2)$
&7.1840 &23.2940 &7.7110 &14.1085 &5.3221 \\
0.5& 0.1& $N(-0.2,0.1^2)$
&7.1840 &12.6431 &6.6113 &10.6372 &4.3573 \\
0.5& 0.2& $N(-0.4,0.15^2)$
&7.1840 &18.3626 &7.0863 &18.8140 &4.3573 \\

\hline
\end{tabular}
\caption{Valuations of European Options under Jump Models}
\label{table:options_ud_jump}
\end{table}

We can see first that under the Vanilla Black-Scholes model with random time default events from Szimayer \cite{Szi2005}, values of both call and put options are smaller than vanilla options, which is consistent with our pricing formula for European options under the random time default model. The default is seen as a random event in the random time default model which is not included in the dynamics of stock, the intensity $\lambda$ here can be seen as a negative compensator similar to the interest rate, thus the option price drops. 

Then we need to focus on the jump-diffusion model, which is more likely to happen in the real market. 
Consider the case there is exactly one negative jump in the holding period, the value of the call option is decreasing and the value of the put option is increasing, because the probability of stock price drops is increasing when compared with the vanilla case (no jump). 
It is consistent with our imaging that with negative jumps, put options will be more valuable and call options will become less valuable.

However, for no more than one jump case and several possible jump cases in the jump-diffusion model, the value of the call option is greater than the vanilla case. It is because the jump compensator $\mu_J$ is big compared with the interest rate and makes the value of no jump call options higher than in the vanilla case (the situation with the biggest proportion, about 70\%-90\%). Furthermore, the cases with no more than one jump actually include about 90\% probability in all situations, together with a relatively small influence of negative jumps to call option value, the value of call options in no more than one jump case and several possible jump cases are very close. 
For put options, the value of options with no more than one jump and several possible jumps are still greater than vanilla put options in most cases, with a larger difference in these two cases compared with call options. The reason for this is the larger influence of negative jumps on put options, with more possible jumps, the value of put options will become higher.

Moreover, with a greater negative mean value of jump size, which means a larger negative jump, the changes in European option prices are greater. 
However, if the jump intensity $\lambda$ increases, which means the frequency of jumps in the same period of time may become higher, the changes of both call and put options decrease.

\subsection{Numerical Study for Hedging Costs under Jump Models} \label{sec6.5.3}

After discussing the valuations of European options under jump models, we now want to see the impact of jumps in the underlying asset on the hedging costs of the standard EPS products. 
Before we do the numerical study, we have some assumptions here. The most important assumption is that the EPS provider will never default to the investors. Furthermore, we focus on the influence of jumps to the hedging costs of the standard EPS products in this numerical study, we do not analyse the potential after-hedge cash flows of the EPS provider because there is no static hedging strategy that can fully hedge the standard EPS products if default events happen. 

We considered the jump-diffusion model with some different situations, firstly, the jump-diffusion model with no more than one jump, which is the general situation in the real market. Then we consider the underlying asset with exactly one jump in the holding period, which can be regarded as a default event if we consider default will always happen at the time of the jump. In the meanwhile, the general jump-diffusion with several possible jumps is included and we assume here no more than 20 jumps inside the holding period of the underlying assets. 
As a remark, we also include the comparable random time default model, it actually is a model that evaluates default risks and it only happens no more than once, so we assume the values of European options have already considered default risks and all financial institutions could default. 
After all, we have the hedging costs called Vanilla hedging costs as another comparison which represents the general case without consideration of jumps.

In this numerical study for default risks, we will choose the general parameters for standard EPS similar to those in Section \ref{sec6.0}, which has been shown in the following Table \ref{table:para_jump}. Moreover, we fix the EPS provider's participation rates to protection legs and fee legs as some constants, because they are less important than the parameters of jumps and default risks.

\begin{table}[h!]
\centering
\begin{tabular}{lc}
\hline
Parameter & Value\\
\hline
Interest rate $r$ & 1.5\% \\
Volatility $\sigma$ & 20\% \\
Initial price of reference portfolio $S_0$ & 100 \\
Time to maturity $T$ & 1 year \\
Protection leg participation $p$ & 0.8 \\
Fee leg participation $f$ & 0.5 \\
\hline
\end{tabular}
\caption{Parameters for Standard EPS under Jump Models}
\label{table:para_jump}
\end{table}

We now get the hedging costs $H(0)$ under several different situations and the results are given in the following Table \ref{table:H0_jump}. Firstly, we have (Van) means Vanilla case without jumps, this is a general fair hedging cost for comparison. The hedging costs under jump-diffusion models are given as ($\leq 1$j) for no more than one jump in the holding period under jump diffusion morel, (1j) for exactly one jump, and (jd) for several possible jumps under the jump-diffusion model (here we use no more than 20 jumps). Apart from that, we have (rt) representing the hedging costs under the random time default model only, without jumps in the underlying asset, as another comparable value.

\begin{table}[h!]
\centering
\begin{tabular}{c|cccc|ccccc}
\hline
\hline
Buffer & $\lambda$ & jump size &$l_1$&$g_1$& $H(0)$(Van) & $H(0)(\leq 1$j) & $H(0)$(1j)& $H(0)$(jd) & $H(0)$(rt) \\
\hline
[1]&0.1& $N(-0.2,0.1^2)$ & -5\% & 5\% & 0.0069 &0.0061 &0.0995 &0.0094 &0.0062 \\

[2]&0.1& $N(-0.2,0.1^2)$ & -5\% & 10\% & 0.0155 &0.0150 &0.1026 &0.0184 &0.0140 \\

[3]&0.2& $N(-0.2,0.1^2)$ & -5\% & 5\% & 0.0069 &0.0037 &0.0888 &0.0120 &0.0056 \\

[4]&0.2& $N(-0.1,0.04^2)$ & -5\% & 10\% & 0.0155 &0.0141 &0.0527 &0.0170 &0.0127 \\

[5]&0.2& $N(-0.2,0.1^2)$ & -5\% & 10\% & 0.0155 &0.0130 &0.0924 &0.0213 &0.0127 \\

[6]&0.2& $N(-0.4,0.15^2)$ & -5\% & 10\% & 0.0155 &0.0109 &0.1719 &0.0381 &0.0127 \\

[7]&0.2& $N(-0.2,0.1^2)$ & -10\% & 10\% & 0.0014 &-0.0006 &0.0668 &0.0061 &0.0011 \\

[8]&0.2& $N(-0.4,0.15^2)$ & -10\% & 10\% & 0.0014 &-0.0017 &0.1391 &0.0216 &0.0011 \\

[9]&0.3& $N(-0.2,0.1^2)$ & -5\% & 5\% & 0.0069 &0.0001 &0.0783 &0.0145 &0.0051 \\
 
[10]&0.3& $N(-0.1,0.04^2)$ & -5\% & 10\% & 0.0155 &0.0124 &0.0480 &0.0177 &0.0115 \\

[11]&0.3& $N(-0.2,0.1^2)$ & -5\% & 10\% & 0.0155 &0.0098 &0.0824 &0.0241 &0.0115 \\
 
[12]&0.3& $N(-0.4,0.1^2)$ & -5\% & 10\% & 0.0155 &0.0041 &0.1501 &0.0503 &0.0115 \\

[13]&0.5& $N(-0.2,0.1^2)$ & -5\% & 10\% & 0.0155 &0.0006 &0.0626 &0.0298 &0.0094 \\
 
[14]&0.5& $N(-0.4,0.1^2)$ & -5\% & 10\% & 0.0155 &-0.0160 &0.1086 &0.0734 &0.0094 \\

\hline
\hline

Floor & $\lambda$ & jump size &$l_1$&$g_1$& $H(0)$(Van) & $H(0)(\leq 1$j) & $H(0)$(1j)& $H(0)$(jd) & $H(0)$(rt) \\
\hline
[1]&0.1& $N(-0.2,0.1^2)$ & -10\% & 5\% & -0.0003 &-0.0032 &0.0473 &-0.0018 &-0.0003 \\
 
[2]&0.1& $N(-0.2,0.1^2)$ & -10\% & 10\% & 0.0083 &0.0057 &0.0504 &0.0072 &0.0075 \\

[3]&0.1& $N(-0.2,0.1^2)$ & -15\% & 10\% & 0.0186 &0.0160 &0.0737 &0.0181 &0.0168 \\
 
[4]&0.2& $N(-0.2,0.1^2)$ & -5\% & 5\% & -0.0144 &-0.0203 &0.0172 &-0.0185 &-0.0118 \\
 
[5]&0.2& $N(-0.2,0.1^2)$ & -10\% & 10\% & 0.0083 &0.0026 &0.0464 &0.0060 &0.0068 \\
 
[6]&0.2& $N(-0.1,0.04^2)$ & -15\% & 10\% & 0.0186 &0.0166 &0.0478 &0.0188 &0.0152 \\
 
[7]&0.2& $N(-0.2,0.1^2)$ & -15\% & 10\% & 0.0186 &0.0128 &0.0683 &0.0177 &0.0152 \\
 
[8]&0.2& $N(-0.4,0.15^2)$ & -15\% & 10\% & 0.0186 &0.0023 &0.0951 &0.0142 &0.0152 \\
 
[9]&0.3& $N(-0.2,0.1^2)$ & -5\% & 5\% & -0.0144 &-0.0239 &0.0139 &-0.0205 &-0.0107 \\
 
[10]&0.3& $N(-0.2,0.1^2)$ & -10\% & 10\% & 0.0083 &-0.0009 &0.0422 &0.0049 &0.0062 \\
 
[12]&0.3& $N(-0.2,0.1^2)$ & -15\% & 10\% & 0.0186 &0.0090 &0.0627 &0.0172 &0.0138 \\
 
[13]&0.3& $N(-0.4,0.15^2)$ & -15\% & 10\% & 0.0186 &-0.0065 &0.0880 &0.0125 &0.0138 \\
 
[14]&0.5& $N(-0.2,0.1^2)$ & -10\% & 10\% & 0.0083 &-0.0094 &0.0328 &0.0026 &0.0051 \\
 
[15]&0.5& $N(-0.2,0.1^2)$ & -15\% & 10\% & 0.0186 &-0.0002 &0.0504 &0.0164 &0.0113 \\
 
[16]&0.5& $N(-0.4,0.15^2)$ & -15\% & 10\% & 0.0186 &-0.0263 &0.0708 &0.0098 &0.0113 \\

\hline
\hline

FC & $\lambda$ & jump size &$l_1$&$g_1$& $H(0)$(Van) & $H(0)(\leq 1$j) & $H(0)$(1j)& $H(0)$(jd) & $H(0)$(rt) \\
\hline
 
[1]&0.1& $N(-0.2,0.1^2)$ & -10\% & 10\% &0.0322 &0.0317 &0.0571 &0.0331 &0.0291 \\

[2]&0.1& $N(-0.2,0.1^2)$ & -15\% & 10\% &0.0425 &0.0420 &0.0804 &0.0441 &0.0384 \\
 
[3]&0.2& $N(-0.2,0.1^2)$ & -10\% & 10\% &0.0322 &0.0309 &0.0545 &0.0341 &0.0264 \\
 
[4]&0.2& $N(-0.1,0.04^2)$ & -15\% & 10\% &0.0425 &0.0419 &0.0604 &0.0440 &0.0348 \\
 
[5]&0.2& $N(-0.2,0.1^2)$ & -15\% & 10\% &0.0425 &0.0411 &0.0764 &0.0458 &0.0348 \\
 
[6]&0.2& $N(-0.4,0.15^2)$ & -15\% & 10\% &0.0425 &0.0382 &0.0979 &0.0496 &0.0348 \\
 
[7]&0.3& $N(-0.2,0.1^2)$ & -5\% & 5\% &0.0181 &0.0167 &0.0276 &0.0194 &0.0134 \\
 
[8]&0.3& $N(-0.2,0.1^2)$ & -10\% & 10\% &0.0322 &0.0299 &0.0519 &0.0352 &0.0239 \\
 
[9]&0.3& $N(-0.2,0.1^2)$ & -15\% & 10\% &0.0425 &0.0398 &0.0723 &0.0475 &0.0315 \\
 
[10]&0.3& $N(-0.4,0.15^2)$ & -15\% & 10\% &0.0425 &0.0364 &0.0922 &0.0543 &0.0315 \\
 
[11]&0.5& $N(-0.2,0.1^2)$ & -10\% & 10\% &0.0322 &0.0276 &0.0464 &0.0375 &0.0195 \\
 
[12]&0.5& $N(-0.2,0.1^2)$ & -15\% & 10\% &0.0425 &0.0367 &0.0639 &0.0512 &0.0258 \\
 
[13]&0.5& $N(-0.4,0.15^2)$ & -15\% & 10\% &0.0425 &0.0330 &0.0790 &0.0650 &0.0258 \\

\hline
\end{tabular}
\caption{Hedging Costs for Standard EPS under Jump Models}
\label{table:H0_jump}
\end{table}

In this table, protection legs and fee legs are given and they have different values that are chosen similarly in the numerical study for standard EPS under normal cases without jumps.
Jump intensities $\lambda$ are set up from ranges $0.1-0.5$, here we need to mention that we use the same $\lambda$ for the jump intensity under the jump-diffusion model in order to represent the default intensity $\gamma^Q$ under the random time default model.
Also, we assume the credit levels of the EPS provider and its counter party are the same here. The hedging costs under the random time default model are actually considered as the situation in that both the EPS provider and its counter party have default risks with $\gamma^c=\gamma^p=\gamma^Q=\lambda$, moreover, both call options and put options have default probabilities here.
In the meanwhile, we still use the normal distribution as jump size and we only consider negative jumps in the holding period such that mean values of jump sizes are negative. Under the assumption of jumps, both the values of call options and put options in the hedging portfolios will be influenced, no matter the directions of the positions.

Firstly, we can see the hedging costs $H(0)$(rt) under the random time default model are all smaller than under normal cases, this is because the values of both call and put options are lower under the random time default model.
In the meanwhile, the hedging costs $H(0)(\leq 1$j) under no more than jump situation are also a little bit smaller than under the Vanilla case, which are close to the hedging costs $H(0)$(rt). Although the hedging costs under the jump-diffusion model should be rising in general situations, the jump compensator will decrease the value of European put options and then the hedging costs. The influence of the jump compensator will be greater compared with jumps under no more than one jump situation because 
jump may not happen in this situation, especially when the jump compensator is large. 

Furthermore, we can see hedging costs $H(0)$(1j) under exactly one jump situation are much higher than under the Vanilla situation without jumps $H(0)$(Van), about 5-10 times. As we already see that the values of call options will decrease and the values of put options will increase under exactly one jump situation, our hedging portfolios contain short positions of call options and altogether long positions of put options, it can be imagined that hedging costs raise when there is exactly one jump in the holding period.

Under a jump-diffusion model with several possible jumps, the hedging costs $H(0)$(jd) are very close to those in the normal case, even compared with the hedging costs $H(0)(\leq 1$j) under no more than jump situation. Hedging costs for buffer EPS are a little bit higher under the jump-diffusion model than under the normal case, and hedging costs for floor EPS are slightly smaller under the jump-diffusion model. This is a very important finding for EPS providers, providers can use options under the jump model to hedge the possible gains or losses, which is much more suitable in the real market.

Because the hedging costs $H(0)$(jd) is close to under normal case without jumps, as a consequence, we can conclude that under the jump-diffusion model with several possible jumps is the best choice for EPS providers to structure hedging strategies when considering no default risk. With larger jump intensity $\lambda$ and all else equal, the hedging costs will become larger. In the meanwhile, the hedging costs $H(0)$(jd) will increase when negative jumps become greater and the mean value of jump size is more negative.

\subsection{Numerical Study for Hedging Costs with Default Probability} \label{sec6.5.2}

In this section, we consider default probabilities for both the counter party and the EPS provider. Here we will consider the underlying asset with jumps that use the jump-diffusion model as the basement. In the meanwhile, we will include the random time default event to represent default risks, the consideration of this default event is less complicated than the jump-diffusion model in which default events are independent of the underlying assets. For convenience, we can use different default intensities $\gamma^c$ and $\gamma^p$ for the counter party and the EPS provider, respectively, in order to present their credit levels.

Considering the analysis of different cases of default risks, here we make some comments: 
\begin{itemize}
\item An important assumption will be applied in the numerical study that the EPS provider will never default to EPS buyers. We have already mentioned that defaulting to investors (and even counter parties) is worthless for pension companies compared with their valuable goodwill. Under this assumption, we can make sure the protections or insurance fees for the EPS buyers still have the same structure as in normal cases and we can focus on the relationship between the EPS provider and the third party to get the fair initial premiums of EPS products.
\item Only the put options will be considered in jump models because we consider negative jumps only. In the meanwhile, the EPS provider always has short positions in call options under the hedging portfolio, defaulting in call options will benefit both the EPS provider and investors.  
\item If we consider the default probability of the counter party, we will set the default intensity $\gamma^c>0$ and $\gamma^p=0$. Oppositely, if we consider the credit risks of the EPS provider, we will set $\gamma^c=0$ and $\gamma^p>0$. The greater the value of the default intensity $\gamma$, the larger the probability of default. In different structures of EPS products, the EPS provider may not have the opportunity to default according to the hedging strategies, for example, buffer EPS.  
\item When we assume all financial institutions have default probabilities, we have both $\gamma^c>0$ and $\gamma^p>0$. However, we will consider $\gamma^c>\gamma^p$ in most cases, because the risk exposure for the counter party is larger than for the EPS provider, considering the positions of put options.
\end{itemize}

We need to focus on the influence of default risks in this section, thus for simplicity, we will use the same single assumption for the jumps under the jump-diffusion model. Moreover, we consider the underlying asset with no more than one jump in the holding period under the jump-diffusion model, the jump is assumed to be negative and relatively large, which is more consistent with the real market in a short period. We set the parameters for jumps to be jump intensity $\lambda=0.2$ and jump size satisfies the normal distribution that $Y_i \sim N(\alpha=-0.2,\delta^2=0.1^2)$. 
Now using the same parameters of standard EPS displayed in Table \ref{table:para_jump}, we can have our numerical study.
By considering different default probabilities of the EPS provider and the counter party, we will be given the following table with the fair initial premiums. We discuss three particular structures of the standard EPS products - buffer EPS, floor EPS, and floor-cap EPS, they have been defined in Section \ref{sec6.0}. Each line in the table represents one particular standard EPS product with different protection legs and fee legs, we introduce 5 buffer EPS, 5 floor EPS and 6 floor-cap EPS.

Here we include the hedging costs under the normal situation (no jumps and default risks) as a comparison, we use Vanilla to represent this situation.  
Moreover, we consider three default cases in this numerical study: only the counter party defaults, only the EPS provider defaults to the counter party (put option buyer), and both of them have default probabilities.
According to our comments for default risks before, we consider the counter party has a higher possibility of default than the EPS provider. So here we have the range of the counter party's jump intensity goes from $0.1$ to $0.5$, and the jump intensity of the EPS provider has a relatively small range from $0.05$ to $0.15$. When considering both of them have credit risks, we will set $\gamma^c>\gamma^p$.

Remember that the EPS provider does not need to short put options in the hedging portfolio for the buffer EPS, which means the EPS provider does not have the chance of defaulting to the counter party in such EPS products, we leave the default intensity of the EPS provider $\gamma^p$ case blank in this situation.

\begin{table}[h!]
\centering
\begin{tabular}{ccc|c|cccc|ccc|c}
\hline
\hline
\multicolumn{3}{c|}{Buffer} &
Van & 
\multicolumn{4}{c|}{Counter party $\gamma^c$} &
\multicolumn{3}{c|}{EPS provider $\gamma^p$} &
$(\gamma^c,\gamma^p)$ \\
 
\hline
No. &$l_1$&$g_1$& $H(0)$ & $0.1$ & $0.2$ & $0.3$ & $0.5$ & -& - & - & - \\
\hline
[1]& -5\% & 5\% & 0.0069 &0.0020 &-0.0017 &-0.0050 &-0.0107 & -& - & - & - \\

[2]& -5\% & 10\% & 0.0155 &0.0111 &0.0075 &0.0041 &-0.0016 & -& - & - & - \\

[3]& -10\% & 5\% & -0.0072 &-0.0107 &-0.0131 &-0.0154 &-0.0192 & -& - & - & - \\

[4]& -10\% & 10\% & 0.0014 &-0.0015 &-0.0040 &-0.0062 &-0.0101 & -& - & - & - \\

[5]& -10\% & 15\% & 0.0080 &0.0057 &0.0032 &0.0010 &-0.0029 & -& - & - & - \\

\hline
\hline

\multicolumn{3}{c|}{Floor} &
Van & 
\multicolumn{4}{c|}{Counter party $\gamma^c$} &
\multicolumn{3}{c|}{EPS provider $\gamma^p$} &
$(\gamma^c,\gamma^p)$ \\
 
\hline
No. &$l_1$&$g_1$& $H(0)$ & $0.1$ & $0.2$ & $0.3$ & $0.5$ & $0.05$& $0.1$ & $0.15$ & $(0.3,0.1)$ \\
\hline
[1]& -5\% & 5\% & -0.0144 &-0.0246 &-0.0298 &-0.0345 &-0.0426 & -0.0168 &-0.0148 &-0.0130 &-0.0305 \\

[2]& -5\% & 10\% & -0.0058 &-0.0155 &-0.0207 &-0.0254 &-0.0335 &-0.0077 &-0.0057 &-0.0038 &-0.0213 \\

[3]& -10\% & 5\% & -0.0003 &-0.0107 &-0.0158 &-0.0205 &-0.0286 &-0.0035 &-0.0022 &-0.0009 &-0.0178 \\
 
[4]& -10\% & 10\% & 0.0083 &-0.0015 &-0.0067 &-0.0114 &-0.0195 &0.0056 &0.0069 &0.0082 &-0.0087 \\

[5]& -15\% & 10\% & 0.0186 &0.0090 &0.0038 &-0.0009 &-0.0090 &0.0156 &0.0164 &0.0172 &0.0008 \\

\hline
\hline

\multicolumn{3}{c|}{Floor-Cap} &
Van & 
\multicolumn{4}{c|}{Counter party $\gamma^c$} &
\multicolumn{3}{c|}{EPS provider $\gamma^p$} &
$(\gamma^c,\gamma^p)$ \\
 
\hline
No. &$l_1$&$g_1$& $H(0)$ & $0.1$ & $0.2$ & $0.3$ & $0.5$ & $0.05$& $0.1$ & $0.15$ & $(0.3,0.1)$ \\
\hline
[1]& -5\% & 5\% & 0.0072 &0.0006 &-0.0045 &-0.0092 &-0.0173 & 0.0085 &0.0104 &0.0123 &-0.0052 \\

[2]& -5\% & 10\% & -0.0014 &-0.0085 &-0.0137 &-0.0184 &-0.0265 &-0.0007 &0.0013 &0.0032 &-0.0143 \\

[3]& -10\% & 10\% & 0.0127 &0.0055 &0.0003 &-0.0044 &-0.0125 &0.0126 &0.0140 &0.0152 &-0.0017 \\
 
[4]& -15\% & 10\% & 0.0230 &0.0160 &0.0108 &0.0061 &-0.0002 &0.0226 &0.0234 &0.0242 &0.0078 \\

[5]& -20\% & 10\% & 0.0299 &0.0234 &0.0182 &0.0135 &0.0054 &0.0296 &0.0301 &0.0306 &0.0145 \\

[6]& -15\% & 15\% & 0.0163 &0.0088 &0.0036 &-0.0011 &-0.0092 &0.0154 &0.0162 &0.0170 &0.0006 \\
 
\hline
\end{tabular}
\caption{Hedging Costs for Standard EPS with Default Risks under the Jump-Diffusion Model}
\label{table:H0_default}
\end{table}

From the Table \ref{table:H0_default}, we can see that the hedging costs will decrease, under the situation that the counter party has default probability $\gamma^c>0$. After considering the credit risks of the counter party, the values of put options sold by the counter party drop down, which decreases the EPS provider's hedging costs of the whole hedging portfolio. In the meanwhile, the price of put options will become lower if the credit levels of the put option seller decrease. Thus when the default intensity of the counter party $\gamma^c$ increases, it means the counter party's default probability rises, the hedging costs will decrease with all else equal. 

Secondly, if we consider the floor EPS and the floor-cap EPS that the EPS provider has default risks, the hedging costs are increasing. Opposite to the counter party default risks, the hedging costs rise when the default probability of the EPS provider increases as the EPS provider needs to sell put options at a lower price.  
Furthermore, the influences of the default intensity of both the counter party $\gamma^c$ and the EPS provider $\gamma^p$ to the floor EPS and the floor-cap EPS are the same, with the same added or reduced amount. Because the structures of the protection legs are the same under these two standard EPS products, with the assumption that call options will never default, we can make the above conclusion. 
Although the default risks from the EPS provider to the counter party give a stimulus to the rising of the hedging costs, the hedging costs under such situations may still lower than those under the Vanilla case, see $\gamma^p=0.05$ and $\gamma^p=0.1$ in Table \ref{table:H0_default}. Remember that we consider the jump-diffusion model in this numerical study, underlying this phenomenon is the fact that the influence of jumps is higher than the default risks of the EPS provider, when the default intensity $\gamma^p$ is small. 

Thirdly, for the buffer EPS, the influence of the counter party's default intensity will become larger if the protection leg $l_1$ decreases (protections of EPS products increase). It can be explained that the EPS provider should purchase more put options so that they can hedge the greater levels of potential protection.
But the protection leg $l_1$ does not affect the influence of the counter party's default intensity for the floor EPS or the floor-cap EPS. This is according to the structure of the protection legs, the EPS provider need to purchase different put options with different strike value in order to hedge the protection from the buffer EPS. However, for the floor EPS and the floor-cap EPS, the put options that were bought from the counter party are the same and the EPS provider will sell different put options for different floor levels. 
Similarly, we may add another conclusion here that the different protection legs $l_1$ in the floor EPS and the floor-cap EPS will make the influence of the default intensities of the EPS provider $\gamma^p$ different. With more protections provided by EPS products (lower the protection legs $l_1$ for the floor EPS and the floor-cap EPS), the influence of the credit levels of the EPS provider becomes smaller. The EPS provider will sell less put options in order to provide higher levels of protection to the EPS buyer, in this case, the influence of default intensity $\gamma^p$ becomes smaller.

When we consider both the counter party and the EPS provider have default probabilities to the other side, here we only have one situation $(\gamma^c,\gamma^p)=(0.3,0.1)$. In this situation, we assume the credit level of the EPS provider is higher than their counter party, which is satisfied our comments before. Then with opposite directions of the influence of default risks from the counter party and the EPS provider, the value of hedging costs will still decrease. First of all, it is because the default intensity of the counter party $\gamma^c$ is higher. In addition, the EPS provider always has a net risk exposure of put options (net long positions of put options) from the counter party in order to hedge their potential protection to the EPS buyers. Thus the influence of the counter party's credit level should be greater than the EPS providers'.

It looks like a contradiction to our previous cognition that the hedging costs would become cheaper rather than more expensive if the counter party has the possibility of default, however, the hedging costs will increase if the EPS provider has default risks. 
Under our assumption, the EPS provider will never default to EPS buyers so that we only discuss the relationship between the EPS provider and the counter party. Actually, the values of the put options we used here have been calibrated for default risks, which means an invisible hypothesis has been applied here that no more default events exist after the calibration. In this situation, the after-hedge cash flows of the EPS provider are always equal to zero under a fair initial premium, and the hedging strategy can fully hedge the potential gains or losses of the EPS provider.
Thus we only consider that the prices of put options containing default risks would be lower, the higher credit risks from the counter party would make the hedging costs smaller. Oppositely, the higher credit risks from the EPS provider would increase the hedging costs.

\subsection{Numerical Study for Default-Adjusted Initial Premiums} \label{sec6.5.5}

From the last section, we analyze the influence of default risks on the hedging costs of the standard EPS products under the jump-diffusion model and we can conclude that the hedging costs with the consideration of default risks are lower than under the Vanilla case. However, it is unfair for the EPS provider that they need to undertake the default risks from the counter party with these lower hedging costs.
Fortunately, although the EPS provider cannot use a static hedge strategy to fully hedge the potential gains or losses, a default-adjusted initial premium has been introduced in order to cover the potential loss after hedging due to default risk and we will do a numerical study on this default-adjusted initial premium now. 

This numerical study can be regarded as an extension of the last section and we will use the same assumptions here. Firstly, only the counter party has default risks that only the EPS provider's long positions of the European put options with consideration of default risks. The predetermined parameters for standard EPS products are given in Table \ref{table:para_jump} and the parameters for the default conditions will be given in the numerical study table. In the meanwhile, we still consider several different cases of jumps, a realistic situation that considers no more than one jump, a simple case that exactly one jump in the holding period of the underlying asset, and the general case with several jumps (here we assume no more than 20, consistent with above section).

\begin{table}[h!]
\centering
\begin{tabular}{c|cc|ccccccc}
\hline
\hline
Buffer & $\lambda$ &$\gamma^c$ & Van & $DA_0$ & $c^D_0$ & $DA_1$ & $c^D_1$ & $DA_n$ & $c^D_n$ \\
\hline
[1]&0.1 & 0.05 & 0.0155 &0.0130 &0.0269 &0.0216 &0.1198 &0.0131 &-0.0466 \\

[2]&0.1 & 0.1 & 0.0155 &0.0254 &0.0372 &0.0421 &0.1349 &0.0255 &-0.0349 \\

[3]&0.1 & 0.2 & 0.0155 &0.0484 &0.0562 &0.0801 &0.1628 &0.0485 &-0.0131 \\

[4]&0.1 & 0.3 & 0.0155 &0.0691 &0.0734 &0.1145 &0.1881 &0.0694 &0.0067\\

[5]&0.1 & 0.5 & 0.0155 &0.1050 &0.1031 &0.1739 &0.2318 &0.1054 &0.0407 \\

[6]&0.2 & 0.1 & 0.0155 &0.0249 &0.0357 &0.0408 &0.1255 &0.0253 &-0.0275 \\

[7]&0.2 & 0.2 & 0.0155 &0.0475 &0.0542 &0.0778 &0.1531 &0.0482 &-0.0063 \\

[8]&0.2 & 0.3 & 0.0155 &0.0679 &0.0710 &0.1112 &0.1780 &0.0689 &0.0129 \\

[9]&0.2 & 0.5 & 0.0155 &0.103 &0.0999 &0.1689 &0.2210 &0.1045 &0.0459 \\

[10]&0.3 &0.1 & 0.0155 &0.0243 &0.0330 &0.0395 &0.1162 &0.0251 &-0.0195 \\

[11]&0.3 &0.2 & 0.0155 &0.0464 &0.0510 &0.0753 &0.1432 &0.0479 &0.0011 \\

[12]&0.3 &0.5 & 0.0155 &0.1007 &0.0954 &0.1634 &0.2097 &0.1039 &0.0518 \\
 
[13]&0.5 &0.2 & 0.0155 &0.0437 &0.0417 &0.0697 &0.1228 &0.0476 &0.0011 \\

[14]&0.5 &0.5 & 0.0155 &0.0949 &0.0835 &0.1514 &0.1860 &0.1032 &0.0518 \\

\hline
\hline

Floor & $\lambda$ & $\gamma^c$ & Van & $DA_0$ & $c^D_0$ & $DA_1$ & $c^D_1$ & $DA_n$ & $c^D_n$ \\
\hline

[1]&0.1& 0.05 & 0.0186 &0.0167 &0.0306 &0.0242 &0.0914 &0.0168 &-0.0410 \\

[2]&0.1& 0.1 & 0.0186 &0.0326 &0.0435 &0.0473 &0.1076 &0.0327 &-0.0260 \\
 
[3]&0.1& 0.2 & 0.0186 &0.0621 &0.0674 &0.0900 &0.1376 &0.0623 &0.0017 \\

[4]&0.1& 0.3 & 0.0186 &0.0888 &0.0891 &0.1288 &0.1647 &0.0890 &0.0268 \\

[5]&0.1& 0.5 & 0.0186 &0.1349 &0.1264 &0.1955 &0.2116 &0.1352 &0.0701 \\

[6]&0.2& 0.1 & 0.0186 &0.0319 &0.0404 &0.0464 &0.1029 &0.0322 &-0.0201 \\
  
[7]&0.2& 0.2 & 0.0186 &0.0608 &0.0636 &0.0883 &0.1329 &0.0613 &0.0066 \\

[8]&0.2& 0.3 & 0.0186 &0.0869 &0.0847 &0.1262 &0.1600 &0.0877 &0.0307 \\
 
[9]&0.2& 0.5 & 0.0186 &0.1319 &0.1209 &0.1917 &0.2067 &0.1331 &0.0724 \\

[10]&0.3& 0.1 & 0.0186 &0.0311 &0.0368 &0.0453 &0.0979 &0.0318 &-0.0141 \\
 
[11]&0.3& 0.2 & 0.0186 &0.0593 &0.0594 &0.0863 &0.1276 &0.0605 &0.0116 \\
 
[12]&0.3& 0.5 & 0.0186 &0.1287 &0.1150 &0.1873 &0.2009 &0.1313 &0.0748 \\
 
[13]&0.5& 0.2 & 0.0186 &0.0560 &0.0494 &0.0816 &0.1155 &0.0592 &0.0221 \\
 
[14]&0.5& 0.5 & 0.0186 &0.1215 &0.1019 &0.1772 &0.1868 &0.1284 &0.0799 \\

\hline
\hline
FC & $\lambda$ & $\gamma^c$ & Van & $DA_0$ & $c^D_0$ & $DA_1$ & $c^D_1$ & $DA_n$ & $c^D_n$ \\
\hline

[1]&0.1& 0.05 & 0.0230 &0.0167 &0.0362 &0.0242 &0.0907 &0.0168 &-0.0020 \\

[2]&0.1& 0.1 & 0.0230 &0.0326 &0.0491 &0.0473 &0.1069 &0.0327 &0.0130 \\

[3]&0.1& 0.2 & 0.0230 &0.0621 &0.0730 &0.0900 &0.1369 &0.0623 &0.0407 \\

[4]&0.1& 0.3 & 0.0230 &0.0888 &0.0947 &0.1288 &0.164 &0.089 &0.0658 \\

[5]&0.1& 0.4 & 0.0230 &0.1130 &0.1143 &0.1638 &0.1886 &0.1133 &0.0885 \\

[6]&0.1& 0.5 & 0.0230 &0.1349 &0.1320 &0.1955 &0.2108 &0.1352 &0.1091 \\

[7]&0.2& 0.1 & 0.0230 &0.0319 &0.0474 &0.0464 &0.1025 &0.0322 &0.0166 \\

[8]&0.2& 0.2 & 0.0230 &0.0608 &0.0707 &0.0883 &0.1324 &0.0613 &0.0433 \\

[9]&0.2& 0.3 & 0.0230 &0.0869 &0.0917 &0.1262 &0.1595 &0.0877 &0.0675 \\

[10]&0.2& 0.5 & 0.0230 &0.1319 &0.1279 &0.1917 &0.2062 &0.1331 &0.1091 \\

[11]&0.3& 0.1 & 0.0230 &0.0311 &0.0453 &0.0453 &0.0978 &0.0318 &0.0205 \\

[12]&0.3& 0.2 & 0.0230 &0.0593 &0.0679 &0.0863 &0.1275 &0.0605 &0.0461 \\
 
[13]&0.3& 0.5 & 0.0230 &0.1287 &0.1236 &0.1873 &0.2008 &0.1313 &0.1094 \\

[14]&0.5& 0.2 & 0.0230 &0.0560 &0.0615 &0.0816 &0.1165 &0.0592 &0.0526 \\

[15]&0.5& 0.5 & 0.0230 &0.1215 &0.1140 &0.1772 &0.1878 &0.1284 &0.1105 \\

\hline
\end{tabular}
\caption{Default Adjusted Initial Premiums for Standard EPS under Jump Models}
\label{table:def_c_jd}
\end{table}

Table \ref{table:def_c_jd} gives us the solutions of the numerical study, the parameters of the jump intensity $\lambda$ for the jump-diffusion model and the default intensity for EPS provider's counter party $\gamma^c$. As a remark, here we assume the jump size to be the same that satisfies the normal distribution that $Y_i \sim N(\alpha=-0.2,\delta^2=0.1^2)$ in the numerical study for default adjusted initial premiums. Because we want to analysis the influence of jumps and default events, we set the fee level $g_1$ and the protection level $l_1$ to be fixed for simplicity. We use $l_1=-5\%$ and $g_1=10\%$ for the buffer EPS and $l_1=-15\%$ and $g_1=10\%$ for the floor and the floor-cap EPS in the numerical study.
WeI’m  have $0$ at the right bottom to represent no more than one jump, $1$ to represent exactly one jump, and $n$ to represent several jumps. We include both expected loss, which has also been discussed in Section \ref{sec6.4.4}, and default-adjusted initial premium $c^D$ under different default events. The fair initial premiums $\widehat{c}$ under the Vanilla case are also provided as benchmarks.

Initially, the default adjustments under an independent default event are very large and make the default-adjusted initial premiums $c^D$ higher than the fair initial premiums under the Vanilla case. The default intensity $\gamma^c$ affects the probability of defaults directly and the greater the default intensity $\gamma^c$ rapidly increases the default adjustments, which has a profound effect on the initial premium compared with jumps. 

Subsequently, comparing three different jump situations, we can see that the default adjustments under exactly one jump situation are the largest, followed by several jumps (here no more than 20 jumps), and then no more than one jump situation. In the meanwhile, the default adjustments $DA_0$ and $DA_n$ are very close, because no jump case actually has the largest proportion (around $90\%$) among these two situations, and this also makes the default adjustments relatively small. Apart from that, the large jump intensity $\lambda$ makes the value of put options increase but the default adjustments decrease (also related to the first protection leg with a participation rate greater than zero $\widehat{l}$).
Consequently, with smaller default intensity $\gamma^c$ and greater jump intensity $\lambda$, the default adjustments $DA$ under an independent default event will be smaller.

Moreover, when we consider the default adjusted initial premium, it cannot be neglected that the default adjusted initial premiums under exactly one jump in the underlying asset situation are very large compared with the fair initial premiums under the Vanilla case. There is a $0.06-0.14$ higher under exactly one jump situation than the fair initial premium under the Vanilla case, for per unit of the nominal principal, according to different values of default intensity $\gamma^c$. In the meanwhile, no more than one jump situation follows the exactly one jump situation, under which the default adjusted initial premiums are also higher than the fair initial premiums. However, with a smaller default intensity $\gamma^c$ of the counter party and a smaller jump intensity $\lambda$, the default adjusted initial premiums would be smaller and closer to the fair initial premium under the Vanilla case, under which the hedging costs decrease more rapidly.

As a remark, the default adjusted initial premium with several jumps $c^D_n$ is the smallest one among the three situations, rather than $c^D_0$ considering default adjustment. Surprisingly, the default adjusted initial premiums under the several jumps situation are lower than the fair initial premiums in some cases, even becoming a negative value. It tells us that the hedging costs under the several jumps situation with independent random time default are very small, and the rising number of jumps makes the hedging costs decrease rapidly. 
Unlike the default adjustments, the default adjusted initial premium $c^D_n$ under the several jumps situation is closer to the fair initial premium $\widehat{c}$ under the Vanilla case when the default intensity for the counter party $\gamma^c$ is relatively high, while $c^D_0$ under no more than one jump case will be less different to the fair initial premium.

Ultimately, we can see that the independent default event will cause a large expected loss. In this case, the default adjusted initial premiums $c^D$ under the independent default event are greater than the fair initial premiums $\widehat{c}$. In order to make this adjustment more meaningful, we should consider a representative situation - no more than one jump when small default intensity $\gamma^c$ and several jumps when high default intensity $\gamma^c$, in which default adjusted initial premiums $c^D$ will be closer to the Vanilla case. Furthermore, we can find some appropriate $n$ - number of jumps - to make the default adjustment initial premium more reasonable.
This conclusion can also be explained in the real market, we can manage both the default intensity and jump intensity. We can see that the default event is less likely to happen in real life and the EPS provider can choose several different counter parties to do risk diversification, they will give the EPS provider a small value of default intensity $\gamma^c$. In the meanwhile, huge negative jumps are actually not very common which makes the jump intensity $\lambda$ small and small jump size and the investor can choose the underlying asset with different characteristics that satisfy their needs. In this case, we can consider the default adjustment initial premium under different jump situations and default intensities, which is consistent with our conclusion.

\section{Conclusion} \label{sec8}

In this paper, the studies of {\it equity protection swaps} (EPS) have been expanded into a more general situation - with consideration of jumps and default risks. An EPS is an insurance product for the variable annuity that provides protection against potential losses and collects insurance fees from actual gains. Compared with guarantees of variable annuities in the U.S. market, EPS has a simpler structure and can be purchased independently of the variable annuity itself. Moreover, it should be a more suitable product for the superannuation market in Australia, with 3.5 trillion Australian dollars (AUD) of total assets, that EPS products are only a kind of insurance product with relatively tiny initial costs compared with nominal principal rather than an annuity ({\it Registered Index-Linked Annuities} (RILAs)). 

According to Xu et al. \cite{Xu2024}, the EPS product can be hedged statically by simply using European call and put options. However, the difference between the valuation time and hedging time cannot be avoided and this time difference makes our pricing based on hedging strategies inconsistent with the real value of the EPS product. Thus we consider random jumps in the underlying asset in order to justify the influence of time difference. Furthermore, such random jumps can simulate the financial crisis that really happens in the market. After carefully analysing, we can make a conclusion that the inclusion of random jumps will not affect static hedging strategies and we get the hedging costs under the jump-diffusion model by using some proper parameters in the numerical study. Because only the valuation of European options will influence the hedging costs of the standard EPS, we start with the pricing of options under the jump model and we can see the influence of jumps to the option values.

We get the main result under the jump-diffusion model that the hedging costs for EPS products with consideration of several possible jumps are very close to the Vanilla situation without jumps. This result could build the EPS providers' confidence in adding financial crisis and the time difference between valuation and hedging into consideration and provide the EPS providers' possibility of resisting more risks. In the real market, EPS providers can only long or short options with existing prices to build hedging portfolios and calibrate the option prices in their own models. We cannot tell the real market model, but it should further consider lots of risks and parameters upon the Vanilla Black-Scholes valuation model. Thus EPS providers can believe in using the jump-diffusion model with jump intensity $\lambda$ under their thoughts, such a jump model will give a hedging cost close to under the Vanilla model with consideration of more possible risks inside.

Furthermore, with random jumps inside the valuation of underlying assets, there is one more important risk that cannot be omitted - credit risks. If a financial crisis happens, the EPS product should provide protection against huge losses which is the significance of introducing insurance. However, we cannot guarantee the actual execution intensity of the EPS provider as they face large losses for the protection, which is related to the credit risks of the EPS provider. In the meanwhile, the EPS provider builds static hedging portfolios to hedge their potential risks and they need to consider the credit risks of their counter parties, especially put option sellers. Moreover, the EPS provider could also be an option seller in regard to their counter parties, related to the structure of EPS products. Altogether, there are credit risks for both EPS providers and their counter parties, we use different default intensities to present different default probabilities of EPS providers and their counter parties.

In this paper, we focus on a more meaningful situation in which the EPS provider will never default to the EPS buyers and only the transactions between the EPS provider and their counter parties have default probabilities. For simplicity, only independent default under Szimayer's \cite{Szi2005} random time default model has been analyzed. From our numerical studies for hedging costs under default risks, we find that the hedging cost will decrease when the counter party has a higher default probability and oppositely, it will increase when the EPS provider has default risks. Although we can find the hedging costs with default risks, the EPS provider cannot fully hedge the credit risks by building the hedging portfolios and such default events with small probability will bring huge losses. 

Fortunately, we can find the expected loss after hedging, assuming there is a default probability, and add this into the hedging cost to get a default-adjusted initial premium. With consideration of counter party's default risks, this default-adjusted initial premium will be greater than the hedging costs and the additional part is the compensation for the EPS provider's credit risk undertaken. However, compared with the fair initial premium under the Vanilla case, the default-adjusted initial premium may not be more expensive because of the default compensator. With such small adjustments to the fair initial premium under the Vanilla case, the EPS provider can include the default risks into consideration and the EPS buyer may even spend less new default-adjusted initial premium. 

As a remark, we need to mention that we only consider the simplest case of default event - independent default - in this work, we may consider some more reliable situations in our further study. It is needless to say that the default event is always related to the financial crisis, financial institutions may be unable to bear huge losses during a financial storm and then the default event happens. So we have another case of the default event, which is related to jumps (we used to represent a financial crisis), that we assume the default happens if the jump size is large enough. Moreover, we need to consider that the European-type options only have one settlement date - maturity, the significant price changes of the underlying asset in the holding period may not really influence the end-of-period closing. Considering the value of options may recover or the option seller has enough time to undertake the losses, we should use the second method to present default risks that the default event is related to the maturity value of the underlying asset. These studies will make the EPS products satisfy the requirements of more general situations and be truly applied to the superannuation market.


\end{document}